%% file: MCQ.tex
\documentclass[10pt,journal,compsoc]{IEEEtran}
\usepackage{graphicx}
\usepackage{balance}
\usepackage{slashbox}
\usepackage[linesnumbered,ruled,vlined]{algorithm2e}
\usepackage{graphicx}
\usepackage{subfigure}
 \usepackage{eqnarray,amsmath}
 \usepackage{amsthm} 
\usepackage{enumitem}
\usepackage[justification=centering]{caption}

\newtheorem{theorem}{Theorem}
\newtheorem{definition}{Definition}
\newtheorem{lemma}{\bf Lemma}[section]

\begin{document}

\title{Quality-Assured Synchronized Task Assignment  in Crowdsourcing}

\author{
	~Jiayang~Tu, ~Peng~Cheng,  ~Lei~Chen,~\IEEEmembership{Member,~IEEE}
	\IEEEcompsocitemizethanks{
		\IEEEcompsocthanksitem The authors are with
		the Department of Computer Science and Engineering, Hong Kong
		University of Science and Technology, Kowloon, Hong Kong, China. Email: jtuaa@cse.ust.hk, pchengaa@cse.ust.hk, leichen@cse.ust.hk.} %
}

\IEEEcompsoctitleabstractindextext{
	
	\input{abstract.tex}
	
	
}

\maketitle
\input{Introduction}

\input{Parameter_Estimation_Model}
\input{Problem_statement}

\input{Greedy_algorithms}
\input{experiment_v1}

\input{Relatedwork_v1}
\input{Conclusion}

\bibliographystyle{abbrv}
\bibliography{added}

\end{document}

%% file: abstract.tex
\begin{abstract}
With the rapid development of crowdsourcing platforms that aggregate the intelligence of Internet workers, crowdsourcing has been widely utilized to address problems that require human cognitive abilities. Considering great dynamics of worker arrival and departure, it is of vital importance to design a task assignment scheme to adaptively select the most beneficial tasks for the available workers. In this paper, in order to make the most efficient utilization of the worker labor and balance the accuracy of answers and the overall latency, we a) develop a parameter estimation model that assists in estimating worker expertise, question easiness and answer confidence; b) propose a \textit{quality-assured synchronized task assignment scheme} that executes in batches and maximizes the number of potentially completed questions (MCQ) within each batch. We prove that MCQ problem is NP-hard and present two greedy approximation solutions to address the problem. The effectiveness and efficiency of the approximation solutions are further evaluated through extensive experiments on synthetic and real datasets. The experimental results show that the accuracy and the overall latency of the MCQ approaches outperform the existing online task assignment algorithms in the synchronized task assignment scenario.
\end{abstract}

%% file: Introduction.tex
\section{Introduction} \label{intro}
No matter in the field of academic research or real-world applications, crowdsourcing has gained significant attention and popularity in the past few years. Serving as a complementary component of human computation, crowdsourcing is leveraged to solve questions that require human cognitive abilities, for example, schema matching \cite{zhang2013reducing} and entity resolutions \cite{wang2013leveraging}. The emergence of multiple well-established public crowdsourcing platforms such as Amazon Mechanical Turk (AMT) and CrowdFlower has facilitated a manageable worker labor market, from which most enterprises can seek services and solutions in a more convenient manner. Although Internet workers of various background can offer joint intelligence, it inevitably brings certain issues due to the difficulty of worker qualification. Workers have different domain knowledge and are error-prone especially when assigned to questions they are unskilled at. To control the quality of answers, a common alternative for most requesters is to design a task assignment scheme that assigns the most beneficial tasks to the target workers.

There are many existing studies to effectively perform \textbf{task assignment} with quality assurance on an online basis \cite{boim2012asking,fan2015icrowd,khan2017crowddqs,liu2012cdas,zheng2015qasca}, where task selection is based on some predefined evaluation metrics such as Accuracy or F-score \cite{zheng2015qasca}, and each coming worker is assigned with the top-$k$ tasks that maximize the evaluation metric. However, online task assignment has its limitation. That is, the number of required repetitions for tasks is set beforehand, thus it is unlikely for the requesters to augment or terminate the allocation of a task according to its answer confidence, which will either cause the return of an uncertain question or an unnecessary waste of worker labor (or a longer delay). In the latter part of this section, we will give a motivating example to further illustrate this problem.

As mentioned above, a task assignment scheme that dynamically determines the number of repetitions of tasks is urgently needed. To improve the existing work, there are two challenges to deal with: 1) an efficient parameter estimation model should be developed. Task assignment is typically studied in isolation or together with parameter estimation in the online scenarios. In other words, the model must be able to keep track of important parameters like the confidence of answers that assist in the task assignment optimization. What's more, this model should run as fast as possible without causing a long delay; 2) the accuracy and the overall latency should be balanced. It is observed that the accuracy of the answers can be improved if we assign more repetitions of tasks to the workers \cite{hirth2013analyzing}, unfortunately, the overall latency or needed cost will consequently increase in this case. Therefore, providing a quality-assured task assignment strategy while not sacrificing the overall latency is of vital importance.

To address the first challenge, we devise a parameter estimation model that assumes \textbf{worker expertise} and \textbf{question easiness} are two potential parameters to collectively determine the \textbf{confidence of answers} \cite{whitehill2009whose}. Specifically, worker expertise denotes the average accuracy of a worker, confidence of answers determines the probability of each answer being the truth and question easiness measures the certainty of current voted answers that guides the question assignment strategy. Intuitively, answers that belong to easier questions and are returned by expert workers will have higher probabilities of being true. Since these three parameters are supposed to be interdependent, the estimation can be accelerated by iterative computation. In addition, the parameters are adaptively updated in the question answering process, so it requires no prior knowledge of workers. Regarding the actual situation of open crowdsourcing platforms where workers are quite dynamic and comprehensive profiles of all workers are hard to obtain, this model is more general to apply in practical crowdsourcing circumstances. In our model, a question is noticed to become easier when the inferred answer is more certain. With this regard, we provide an easiness score threshold $\delta$ denoting the certainty of returned answers, and requesters can set the threshold value according to their preferences to control the accuracy of answers. 

To resolve the second challenge, based on the parameter estimation model stated above, we run our task assignment scheme batch by batch \textit{(synchronized task assignment)} and renew the model parameters at the end of each batch. At the very beginning of every batch, we carefully consider simultaneous task processing situation and assign each idle worker to one task, aiming at maximizing the number of completed questions without wasting worker labor. In order to fulfill this objective, workers are assigned to a question set $Q_{T}$ that is as packed as possible, namely, $|Q_{T}|$ is minimized. The advantages of batch processing and the optimization goal are two-fold: 1) worker labor is concentrated on the smallest set of questions, then for each batch it is likely to gain more completed questions, thus the overall latency can be controlled; 2) fewer questions are occupied for worker processing, which reserves more available questions for next batch workers and enables further optimization. The latter advantage explains the reason why batch processing is adopted in this paper.

As we mentioned earlier, current task assignment assumes an online task assignment scenario, that is, tasks are set with a fixed number of repetitions beforehand and assigned to the workers once they join the platform. Such settings can have the following problems: (1) uncertain answers are likely to return if the number of repetitions is limited, resulting in less reliable data quality; (2) short-sighted task assignment may happen due to the unconsciousness of the overall situation, which may cause exceeding assignment of the questions, waste worker labor and lead to longer delay. To improve the above problems, in some real applications like Didi Chuxing \cite{Zhang:2017}, batch by batch assignment (synchronized task assignment) \cite{Zhang:2017,haas2015clamshell} is widely utilized in which the joint benefit of multiple workers is optimized. Next, we show the difference between synchronized task assignment and top-$k$ online assignment regarding problem (1) with a concrete example as follows.

\textbf{Motivating Example.} \textit{Suppose two workers $W=\{w_{1}, w_{2}\}$ arrive sequentially in the question pool $Q=\{q_{1}, q_{2}\}$. And we judge whether a question can be returned by verifying whether its easiness score has reached the threshold $\delta$. Intuitively, easiness score measures the current certainty of the voted answers regarding a question and is calculated from confidences of answers, while $\delta$ is the requirement of the question certainty set by the requesters. Assume $\delta$ and the current estimated easiness score of $q_{j}$, denoting by $sc.d(q_{j})$ are known beforehand, thus the remaining easiness score of $q_{j}$  is equal to $c_{j} = \delta - sc.d(q_{j})$. In this example, we assume $c = \{0.35, 0.25\}$. Besides, the expected easiness score increase of all workers to all tasks is shown in TABLE \ref{mov}.} 
\begin{table}[t]
	\centering
	\caption{Easiness Score Increase}\label{mov}
	\bigskip
	\vspace{-4ex}
	\begin{tabular}{|c|c|c|}
		\hline
		\backslashbox{$W$}{$Q$} & $q_{1} (c_{1} = 0.35$, $rep_{1} = 0)$ & $q_{2} (c_{2} = 0.25$, $rep_{2} = 2)$\\
		\hline
		$w_{1}$ & $0.2$ & $0.2$\\
		\hline
		$w_{2}$& $0.1$ & $0.1$\\
		\hline
	\end{tabular}\vspace{-2ex}
\end{table}

\textit{\textbf{a) Synchronized task assignment.} In synchronized task assignment, we aim at making the most efficient utilization of worker labor and minimize $|Q_{T}|$ (the number of assigned tasks), so the remaining easiness score of questions $c$ cannot be overused. We show 4 possible cases of question assignment as the following:}
\vspace{-1ex}

\hspace{3ex}\textbf{Case 1:} $w_{1}, w_{2} \rightarrow q_{1}$: feasible and $|Q_{T}|=1$

\hspace{3ex}\textbf{Case 2:} $w_{1} \rightarrow q_{1}$, $w_{2} \rightarrow q_{2}$: feasible and $|Q_{T}|=2$

\hspace{3ex}\textbf{Case 3:} $w_{1}, w_{2} \rightarrow q_{2}$: overuse

\hspace{3ex}\textbf{Case 4:} $w_{1} \rightarrow q_{2}$, $w_{2} \rightarrow q_{1}$: feasible and $|Q_{T}|=2$

\vspace{1ex}

\textit{In the 4 above cases, case 1 is a feasible assignment because 0.35 - 0.2 - 0.1 $>$ 0 (similar with case 2 and 4) and case 3 is not adopted since it overuses the remaining easiness score of $q_{2}$ (0.2 + 0.1 $>$ 0.25) and results in a waste of worker labor. According to the optimization goal of synchronized task assignment, we tend to select case 1 to minimize $|Q_{T}|$. And two repetitions of $q_{1}$ are created immediately and assigned to $w_{1}$ and $w_{2}$. }

\textit{\textbf{b) Online task assignment} (assume $k = 1$). According to online task assignment, workers arrive in the order of $w_{1}$, $w_{2}$, and each of them should be allocated with one question. For each worker, assume the evaluation metric is to assign her to one question whose remaining easiness score is minimized. As noted in Table \ref{mov}, suppose the number of repetitions left for $q_{1}$, $q_{2}$ are 0, 2 respectively (denoted as $rep_{1} = 0 $ and $rep_{1} = 2$) and $c = \{0.35, 0.25\}$ stays the same. The reason why $rep_{1} = 0$ but $c_{1}$ still has a large value is that the answers of $q_{1}$ contributed by previous workers are so diverse that the remaining easiness score (uncertainty) is kept high. Since $rep_{1} = 0 $, the only task assignment method is case 3 ($w_{1}, w_{2} \rightarrow q_{2}$). As a result, $q_{1}$ is returned but far from certainty while the assignment has exceeded the remaining easiness score $c_{2}$ of $q_{2}$ and wastes worker labor.}

The motivation example reveals that compared to top-$k$ online task assignment, synchronized task assignment is able to make efficient utilization of worker labor. We keep track of the remaining easiness score of questions to carefully assign workers and decide whether to increase or decrease the allocation of a question. Therefore, the number of repetitions of every question can be determined on the fly and unnecessary expenses can be avoided for the requesters. And answers can be returned with quality assurance.

To sum up, we make the following contributions:
\begin{itemize}[leftmargin=*]
	\item We develop an efficient parameter estimation model that assists in estimating worker expertise, question easiness and answer confidence. 
	\item We propose a quality-assured synchronized task assignment scheme executing in batches and \textit{maximizing the number of completed questions} (MCQ) in every batch. And we prove that MCQ problem is NP-hard. 
	\item We present two efficient approximation algorithms to address the MCQ problem. Extensive experiments are conducted on synthetic and real datasets to evaluate them.
\end{itemize}

The rest of this paper is organized as follows. In Section \ref{model}, we introduce our parameter estimation model. And we formally formulate the MCQ problem in Section \ref{ps}. Then we present and analyze two greedy solutions in Section \ref{solution}. We evaluate the performance of the algorithms in Section \ref{experiment}. Related literature is discussed in Section \ref{relatedwork}. We finally conclude the work in Section \ref{conclusion}.

%% file: Parameter_Estimation_Model.tex
\section{Parameter Estimation Model} \label{model}
In this section, we begin with the modeling of basic crowdsourcing components and show the detailed development of the parameter estimation model. We lay the foundation of the optimization problem in next section and briefly introduce several standard crowdsourcing concepts commonly used throughout this paper.
\begin{itemize}[leftmargin=*]
	\item{\textbf{Worker.}} Workers arrive at the crowdsourcing platfom and select human intelligence tasks to peform. In this paper, workers are assumed to be rational, meaning that there exists no group of malicious workers that aims to dominate the answers and destroy the benifits of the crowdcourcing platform or the other workers.
	
	\item{\textbf{Question.}} Questions (or tasks) are created and published to crowdsourcing platforms by requesters. The task types are usually varying from multiple-choice questions to numeric questions, both of which are micro-tasks that require short processing time for workers.
	Note that \textit{questions} and \textit{tasks} are interchangeably used in this paper and they refer to the same entity.
	\item{\textbf{Repetition.}} Repetition is the most atomic unit of a question that a worker can respond to. The number of repetitions regarding a question is declared by the requester on task creation with the purpose of improving the answer reliability. For a single question, the repetitions are identical and are assigned to different workers. 
	\item{\textbf{Answer.}} Answers are submitted worker votes. Answers to the same question are often various due to the distinct worker expertise. Note that \textit{answers}, \textit{choices} and \textit{votes} are interchangeably used in this paper and they refer to the same entity.

\end{itemize}

\input{Notations}
\subsection{Fundamental Definitions}
We first give three definitions: Answer Confidence, Question Easiness and Worker Expertise that mainly compose the parameter estimation model and play an important role in the synchronized task assignment.
\begin{definition}[\textit{\textbf{Answer Confidence}}]\label{def1}
The confidence of an answer $a$ of question $q$, denoted as $c(a)$, is the probability that $a$ is the truth of $q$.
\end{definition}

\begin{definition}[\textbf{\textit{Question Easiness}}]\label{def2}
The easiness of a question $q$, denoted as $d(q)$, is the adjusted average of the total pairwise distances between confidences of different answers. $0 < d(q) < 1$, where a larger value of d(q) indicates that $q$ is more certain to return the answer with a higher confidence as the truth.
\end{definition}

\begin{definition}[\textbf{\textit{Worker Expertise}}]\label{def3}
The expertise of a crowdsourcing worker $w$, denoted as $e(w)$, is the average expected confidence over all answers $w$ has voted.
\end{definition}

Intuitively, without ground truths, Answer Confidence helps estimate the probability of each answer being the truth; Question Easiness is used to measure how far a question is away from certainty in order to guarantee the quality of the returned answers; Worker Expertise denotes the probability of a worker that can provide correct answers for questions.

\begin{figure}[t] \centering
	\includegraphics[scale=0.5]{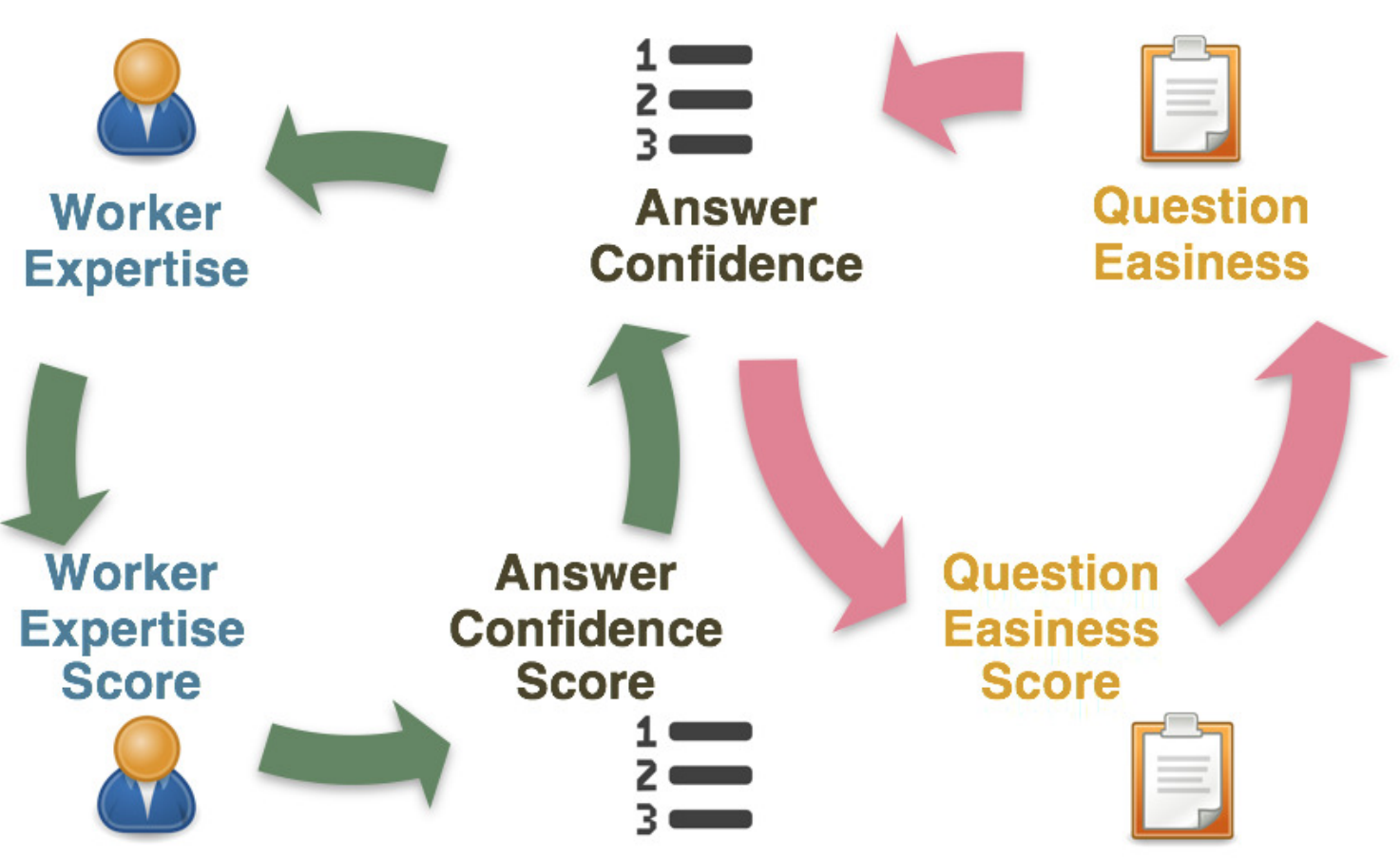}   
	\bigskip     \vspace{-3ex}
	\caption{The Dual-Cycle Parameter Estimation Model}\label{fig:model}
	\vspace{-4ex}
\end{figure}

\vspace{-1ex}
\subsection{Iterative Parameter Estimation}
The iterative parameter estimation model consists of two inference circles in Figure \ref{fig:model}, namely the inference between worker expertise and answer confidence (the green cycle) and the inference between question easiness and answer confidence (the red cycle).

As discussed in Section \ref{relatedwork}, TRUTHFINDER is run on aggregated answers and performs well on conflicting web information that usually has a large number of votes (more than 5) regarding an identical object. However, in this work, we focus on paid crowdsourcing tasks where the number of needed repetitions for a question should be as less as possible to gain equally reliable estimated parameters, for example, answer confidence. And the main challenge is that when the votes are too few and diverse at the same time, we cannot directly utilize the power of the TRUTHFINDER prototype to resolve. Together with the requirement of tiny delay in the synchronized task assignment scenario, even with very few votes and they are quite diverse, the parameter estimation model should still have the ability to act efficiently and accurately to calculate the parameters.

\subsubsection{Inference between Worker Expertise and Answer Confidence}
As given in Definition \ref{def3}, the \textit{\textbf{expertise of a worker}} $w$ can be expressed as: 
\begin{equation}
\begin{split} \label{expertise}
e(w) = \frac{{\sum_{a \in A(w)}}c(a)}{|A(w)|}, \quad e(w) \in [0, 1)
\end{split}    
\end{equation}
where $A(w)$ refers to the set of answers provided by $w$.

Given the question easiness $d(q)$, the confidence of the answer $a$ is calculated as $c(a) = d(q)(1 - \prod_{w \in W(a)}(1 - e(w))$, where $W(a)$ refers to the set of workers offering the same answer $a$ to question $q$. The formula is not hard to understand. $\prod_{w \in W(a)}(1 - e(w))$ is the probability that no workers assigned to $q$ choose $a$, so $c' = 1 - \prod_{w \in W(a)}(1 - e(w))$ represents the probability that $a$ is selected. However, in this work, we think that answer confidence does not merely rely on the expertise of the workers but also is determined by the question nature itself, namely, the question easiness. Apparently, if a question is easier, the answer obtained from the same worker will have a higher possibility to be the truth, and this is the reason why we yield $c$ by multiplying $c'$ by a factor of $d(q)$.
However, to avoid the underflow situation in $c(a)$ resulted from the repetitive multiplication of the term $(1 - e(w))$, we further define the \textit{\textbf{expertise score of a worker}} $w$, denoted as $sc.e(w)$, as followings:
\begin{equation}
\begin{split}\label{expertise score}
sc.e(w) = -\ln(1 - e(w))
\end{split}    
\end{equation}
Similarly, for the ease of calculating answer confidence from the expertise score, \textit{\textbf{confidence score of an answer}} $a$, denoted as $sc.c(a)$, is defined as: 
\begin{equation}
\begin{split} \label{confidence score}
sc.c(a) = -\ln(1 - \frac{c(a)}{d(q)})
\end{split}    
\end{equation}

We derive the relationship between expertise score and confidence score as $sc.c(a) = \sum_{w \in W(a)} sc.e(w)$ according to Equation \ref{expertise} - \ref{confidence score}. And the deduction process is shown in the following:\vspace{-2ex}
\begin{equation}
\begin{split}
c(a) &= d(q)(1 - \prod_{w \in W(a)}(1 - e(w))\\
\Leftrightarrow \ \ \ \ \ 1-\frac{c(a)}{d(q)}&	= \prod_{w \in W(a)}{(1 - e(w))}\\
\Leftrightarrow  \ln(1-\frac{c(a)}{d(q)}) & = \sum_{w \in W(a)}{\ln(1 - e(w))}\\
\Leftrightarrow \ \ \ \ \ \ \  sc.c(a) &= \sum_{w \in W(a)} sc.e(w)\\
\end{split}
\end{equation}

Next, we update the expression of \textit{\textbf{answer confidence}} with the help of Equation \ref{confidence score}: 
\begin{equation}
\begin{split} \label{deduct}
&\mathrm{e}^{sc.c(a)} =\mathrm{e}^{ -\ln(1 - \frac{c(a)}{d(q)})}\\
\Leftrightarrow \quad & c(a) =  \frac{d(q)}{1-\mathrm{e}^{-sc.c(a)}}\\
\end{split}    
\end{equation}
In Equation \ref{deduct}, answer confidence $c(a)$ can be negative, which is not reasonable according to our definition. In order to guarantee $c(a)$ has a positive value, we apply a logistic function and Equation \ref{deduct} is adjusted to $c(a) = \frac{d(q)}{1+\mathrm{e}^{-sc.c(a)}}$, where $c(a)$ is mapped to a smaller range $(0,1)$ \cite{yin2008truth}. This adjustment can be intuitively interpreted: as the easiness of the question increases (higher $d(q)$), $c(a)$ is larger, meaning that the answer $a$ has a higher probability to be the truth; when the overall expertise score of workers who vote $a$ is higher ($sc.c(a)$ becomes higher since $sc.c(a) = \sum_{w \in W(a)} sc.e(w)$), then $a$ also has a higher probability of being correct.

All $c(a), a \in A(q)$ is normalized after the calculation is finished, which ensures the positive values of worker expertise. The inference between worker expertise and answer confidence is depicted as the green cycle in Figure \ref{fig:model}.

\vspace{-0.5em}
\subsubsection{Inference between Question Easiness and Answer Confidence}
As suggested above, the answer confidence $c(a)$ is adjusted to $\frac{d(q)}{1+\mathrm{e}^{-sc.c(a)}}$. If we set $d(q)$ to 1, the inference between worker expertise and answer confidence is exactly the TRUTHFINDER framework. However, as discussed in the Section \ref{relatedwork}, TRUTHFINDER performs well simply when the question has a substantial number of votes. Based on its framework, we propose the question easiness concept in Definition \ref{def2} and develop another important inference cycle between answer confidence and question easiness, which is depicted as the red cycle in Figure \ref{fig:model}.

Let $A(q)$ be the set of all possible answers of $q$. Then the number of pairs within $A(q)$ is $p = \frac{|A(q)|(|A(q)|-1)}{2}$. In addition, we define the \textit{\textbf{easiness score of a question }}$q$ (denoted by $sc.d(q)$) as follows: 
\begin{equation}
sc.d(q) = \frac{\sum_{a_{i}, a_{j} \in A(q), a_{i} \neq a_{j} }{|c(a_{i}) - c(a_{j})|}}{p}
\end{equation}

The confidence differences over all answer pairs are summed up and averaged with the motive to measure the certainty of the voted answers. For example, given a binary-choice question with two answers $a_{1}$ and $a_{2}$ and there are two cases of the confidence distribution: (1) $c(a_{1}) = c(a_{2}) = 0.5$; (2) $c(a_{1}) = 0.8, c(a_{2}) = 0.2$. Obviously, case 2 is certain enough to return $a_{1}$ while case 1 is far from certainty. In other words, a large confidence difference indicates high certainty.

However, it is noticed that $sc.d(q)$ can vary from 0 to 1, if we pass the value of $sc.d(q)$ to $d(q)$ in the formula of $c(a)$, $c(a)$ is quite sensitive to the variations of $d(q)$. To smooth the fluctuation of $sc.d(q)$, a new concept \textbf{\textit{easiness}} ($d(q)$ in Definition \ref{def2}) is introduced and we set $d(q) = \frac{1}{1+k \cdot \mathrm{e}^{-sc.d(q)}}$, where $k \in (0,1]$ and $d(q)$ is thus mapped to $[\frac{1}{1+k}, 1)$.

The advantages of introducing question easiness are summarized into two aspects: a) For those difficult questions with uncertain answers, the confidence of their answers will be diminished because of the effect of $d(q)$ in the expression of $c(a)$. Therefore, worker expertise or expertise scores that originally dominated by these questions are declined. Consequently, the difficult questions caused by few and uncertain answers tend to become easier since the confidence of the answers voted by high-expertise workers starts to increase and exceed that of others. Due to this reason, the parameter estimation can be accelerated; b) Since question easiness is defined to measure the certainty of the tasks, it is beneficial for us to know which questions are demanding additional reliable workers. Therefore, our parameter estimation model is applicable to crowdsourcing circumstances where the worker answers are very few and diverse. Besides, it acts as a helpful sign to dynamically guide the task assignment in the latter section.

\subsubsection{Execution of the Dual-cycle Parameter Estimation Model}
The whole parameter estimation model is demonstrated in Figure \ref{fig:model}. As described above, we can estimate the worker expertise and question easiness if we know the answer confidence, and vice versa. We adopt an iterative method to run the model \cite{yin2008truth}. At the end of each batch, the model is started to calculate the parameters according to the worker answers. First of all, we initialize the expertise of workers to be equal, and we run the left cycle for a fixed number of iterations. And starting from answer confidence, we run the right cycle for the same fixed number of iterations. Finally, this dual-cycle estimation is repeated until answer confidence varies in a small range. The details are studied in Section \ref{experiment}.

%% file: Notations.tex
\begin{table}	\label{table:notation}
	\centering \small{
	\begin{footnotesize}
		\caption{Summary of Symbols}\vspace{-4ex}
		\bigskip
		\begin{tabular}{c|l}

				Symbol &  \qquad\qquad\qquad Description \\
				\hline\hline
				$c(a)$ & the confidence of answer $a$\\
				$d(q)$ & the easiness of question $q$\\
				$e(w)$ & the expertise of worker $w$\\
				$A(w)$ & the set of answers provided by $w$\\
				$W(a)$ & the set of workers offering answer $a$\\
				$A(q)$ & all possible choices of question $q$\\
				$l$ & the number of answers of each question\\
				$t$ & the truth of a question \\
				$\delta$ & the easiness score threshold \\
				$w_{i}$ & the $i^{th}$ worker \\
				$q_{j}$ & the $j^{th}$ question \\
				$c_{j}$  & the remaining easiness score of question $q_{j}$\\
				$Q_{T}$ & the set of questions assgined to workers in 
				each batch\\
				$I$ & the question assignment scheme in each batch\\
				$E^{2}I_{ij}$ & the expected easiness score increase of $w_{i}$ to $q_{j}$\\
				$W(q_{j})$ & the set of workers assigned to quesiton $q_{j}$\\
				$Q$ & the set of questions in this batch\\
				$W$ & the set of available workers in this batch\\
				$Q_{o}$ & the set of current open questions\\
				$U_{i,j}$ & the remaining easiness score of $q_{j}$ after answered  by $w_{i}$\\
				$X_{j}$ & the set of workers assigned by XM Greedy to $j$th question \\
				\hline
		\end{tabular}\vspace{-3ex}
	\end{footnotesize}}
\end{table}

%% file: Problem_statement.tex
\vspace{-3ex}
\section{Problem Statement} \label{ps}
\vspace{-1ex}
In this section, we first introduce some preliminary definitions, and then formally propose the \textit{synchronized task assignment to maximize the number of completed questions} (MCQ) problem. Finally, we prove that MCQ problem is NP-hard.

\begin{definition}[\textit{\textbf{Easiness score Increase}} (\textit{\textbf{$EI$}})]
Given a question $q$ with $l$ answers ($l \geq 2$), let $W(a_{i})$ denote the sets of workers who have voted $a_{i}$. As defined in Section \ref{model}, the number of answer pairs is $p = \frac{l(l-1)}{2}$ and thus the easiness score of $q$ is $sc.d(q) = \frac{\sum_{a_{i}, a_{j} \in A(q), a_{i} \neq a_{j} }{|c(a_{i}) - c(a_{j})|}}{p}$. Pick a random answer $a_{k}$, $c(a_{k}) = \frac{d(q)}{1+\mathrm{e}^{-sc.c(a_{k})}}$ where $sc.c(a_{k}) = \sum_{w \in W(a_{k})} sc.e(w)$. Assume a coming worker $w$ votes $a_{k}$, then $EI$ incurred by worker $w$ regarding question $q$ is expressed as $EI(A(w) = a_{k}|q) = sc.d(q)'-sc.d(q)$. Note that the calculation of $sc.d(q)'$ is similar with $sc.d(q)$, except that in $sc.d(q)'$, $c'(a_{k}) = \frac{d(q)}{1+\mathrm{e}^{-sc.c(a_{k})'}}$ where $sc.c(a_{k})' = sc.c(a_{k}) + sc.e(w)$.
\end{definition}

$EI$ can vary from negative to positive values. Specifically, when the answer $a_{k}$ above reduces the average confidence distance among all answers, the easiness score will decrease, then $EI(A(w) = a_{k}|q)= sc.d(q)'-sc.d(q) < 0$, and vice versa. The definition of $EI$ helps in measuring the benefit of a worker $w$ to a question $q$, where the benefit is evaluated by the expected increase of the easiness score incurred by $w$ regarding $q$ as follows.

\begin{definition}[\textit{Expected Easiness score Increase (\textbf{$E^{2}I$})}] 
	The $E^{2}I$ of a worker $w$ regarding question $q$ is: \vspace{-2ex}
		\begin{equation} \label{eei}
		E^{2}I(w|q) = |\sum_{k = 1}^{l} P(A(w) = a_{k}|q) \cdot EI(A(w) = a_{k}|q)|
		\end{equation}
\end{definition}

Equation \ref{eei} calculates the expectation of the easiness score increase ($EI$) that a worker brings to a question. $P(A(w) = a_{k}|q)$ means the probability of worker $w$ giving the answer $a_{k}$ to question $q$, and $EI(A(w) = a_{k}|q)$ is the corresponding $EI$ when $w$ votes $a_{k}$. For each possible answer, these two terms are multiplied and accumulated, finally we get the absolute value of the result. The reason why we take the absolute value is that $EI(A(w) = a_{k}|q)$ can be negative when the $sc.d(q)$ decreases. Without knowing the correct answer of $q$, the answer that shows its advantage in determining the absolute sum, namely $E^{2}I(w|q)$, is regarded as the most possible answer given by $w$. Therefore, $E^{2}I$ represents a fair expectation measurement for every worker since no default truth is assumed.

Next, we show how to calculate the term $P(A(w) = a_{k}|q)$ in Equation \ref{eei}. Let $t$ denote the true answer of question $q$, then Equation \ref{eei} can be calculated as the following:
\begin{equation}
P(A(w) = a_{k}|q) = \sum_{r=1}^{l} P(A(w) = a_{k}|t = r) \cdot P(t = r),
\end{equation}
where $P(t=r)$ is probability of answer $r$ being the truth, which is exactly the answer confidence $c(r)$. And in order to find $P(A(w) = a_{k}|t = r)$, we assume the worker selection probability equals to $p = \frac{1}{1+\mathrm{e}^{-e(w) \cdot d(q)}}$ \cite{whitehill2009whose}, which means that worker expertise and question easiness can jointly affect worker $w$ to select a correct answer to question $q$. Therefore, we have: $P(A(w) = a_{k}|t = r) = p$, if $a_{k}= r$, meaning the worker answer $a_{k}$ happens to be the truth $r$; otherwise, $P(A(w) = a_{k}|t = r) = \frac{1-p}{l-1}$, that is, worker has a equal probability $\frac{1-p}{l-1}$ to choose a wrong answer.

As mentioned in Section \ref{intro}, the \textbf{easiness score threshold} $\delta$ is provided for the requester to control the accuracy of answers and reduce their cost. Questions with easiness score equal to or exceeding $\delta$ will be returned from the question pool at the end of each batch. Therefore, if both the current easiness score $sc.d(q)$ of question $q$ and $\delta$ are known, the remaining easiness score of $q$, denoted by $c$, is equal to $\delta - sc.d(q)$.

Then we formally propose our optimization problem as the following.
\begin{definition}[\textit{\textbf{the Maximizing Completions of Questions (MCQ) Problem}}] \label{mcq}
	Given a worker set $W = \big\{w_{1}, w_{2}, \dots, w_{n}\big\}$ of size $n$ and a question set $Q = \big\{q_{1}, q_{2}, \dots, q_{m}\big\}$ of size $m$ in the current batch. Let $c_{j}$ denote the remaining easiness score of $q_{j}$ and $E^{2}I_{ij}$ denote the expected easiness score increase of worker $w_{i}$ to question $q_{j}$. To maximize the number of completed questions and make the most efficient utilization of the worker labor, our task assignment scheme $I$ aims to assign each worker to \textit{one} question in order that \textit{all workers can be packed into a question set $Q_{T}$ and $|Q_{T}|$ is minimized}. Note that minimizing $|Q_{T}|$ means to tightly assign all available workers to the least number of questions in order that the most number of questions can be finished in each batch.
\end{definition}

\noindent\textbf{NP-hardness of the MCQ problem.} We prove that the MCQ problem is NP-hard by reducing it from the Bin Packing problem \cite{michael1979computers}. 
\begin{theorem}\label{theorem:nphard]}
	The MCQ problem is NP-hard. 
\end{theorem}
\begin{proof}
	Specifically, we show the proof in detail by a reduction from the Bin Packing problem. A Bin Packing problem can be described as follows: given a set of bins $s_{1}, s_{2}, \dots, s_{m}$ with the same size $V$ and a list of $n$ items with size $a_{1}, a_{2}, \dots, a_{n}$, the problem is to find the minimum number of bins $B$ and a B-partition $s_{1} \cap s_{2} \cap \dots \cap s_{B}$  \textit{s.t.} $\sum_{i \in S_{k}} a_{i} \leq V$.

	For a given Bin Packing problem instance, we can transform it to an instance of MCQ as follows: we set $E^{2}I_{ij} = a_{i}$ and $c_{j} = V$, which means workers and questions are treated as items and bins accordingly, and workers have the same $E^{2}I$ over all questions. According to the MCQ problem, our target is to assign all workers to questions in order that the number of assigned questions is minimized, which is equivalent to minimize the number of used bins of equal size under the constraint that all items should be packed into bins. Given this mapping, it is easy to show that Bin Packing Problem instance can be solved if and only if the transformed MCQ problem can be solved. Therefore, the optimal solution of the MCQ problem reveals the optimal solution of the Bin Packing Problem. Therefore, the MCQ problem is NP-hard.
\end{proof}

%% file: Greedy_algorithms.tex
\vspace{-2ex}
\section{Assignment Algorithms} \label{solution}
In this section, we propose two greedy algorithms to address the MCQ problem and study their approximation ratios. Before that, we first illustrate some preliminary concepts.

\subsection{Preliminary Concepts}
We now define the following terms which will be commonly used in the approximation algorithms. 
\begin{algorithm}[t]

		\caption{\small{First Match Algorithm}}
		\label{alg:fm}
		\KwIn{$E^{2}I$, $W$, $Q$ and $C=c_{1}, c_{2},\ldots, c_{n}$}
		\KwOut{$Q_{T}$, question assignment scheme I}
		\For {$w_{i} \in W$}
		{
			\For {$q_{j} \in Q$    }
			{
				Let $U_{i,j} = c_{j} - E^{2}I_{i,j}$, the remaining $c_{j}$ after $w_{i} \rightarrow q_{j}$;
			}
		}
		Find $U_{i',j'}$ = min $U_{i,j}$, $\forall w_{i} \in W, \forall q_{j} \in Q$;\\
		Let $w_{i'}$ be $w_{1}$, $I = \{w_{1} \rightarrow q_{j'}\}$; \\
		Let $Q_{o}$ denote the set of open questions sorted by the insertion order;\\
		$Q_{o} = \{q_{j'}\}$;\\
		Randomly order the workers from $w_{2}$ to $w_{n}$;\\
		$\forall i > 1$, update $U_{i,j'} = U_{i,j'} - E^{2}I_{1,j'} $;\\
		\For {$w_{i} \in W-\{w_{1}\}$}
		{
			\eIf{$\forall q_{j} \in Q_{o}$, $U_{i,j} < 0 $}{
				Let $U_{i,k}$ = min $U_{i,j}$, $\forall q_{j} \in Q-Q_{o}$;\\
				Open the closed question $q_{k}$, set $Q_{o} = Q_{o} \cup \{q_{k}\}$; \\
			}{
				Find the first $q_{k} \in Q_{o}$ \textit{s.t.} $U_{i,k} \geq 0$;\\
			}
			$I = I \cup \{w_{i} \rightarrow q_{k}\}$;\\
			$\forall z > i$, update $U_{z,k} = U_{z,k} - E^{2}I_{i,k} $;\\
		}    
		$Q_{T} = Q_{o}$  
\end{algorithm}

\begin{itemize}[leftmargin=*]
    \item{\textbf{Feasible task assignment.}} A question assignment scheme $I$ is called feasible if $\forall q_{j}$, the total $E^{2}I$ of its assigned workers $W(q_{j})$ does not exceed its remaining easiness score, which means $\sum_{w_{i} \in W(q_{j})} E^{2}I_{ij} \leq c_{j}$.
    \item{\textbf{Open/closed questions.}} A question is called open when it is assigned to workers and is called closed if no workers are assigned to it.
\end{itemize}

Next, we introduce two algorithms, namely, the \textit{First Match Algorithm} and \textit{Best Match Algorithm}. And their time complexities are further analyzed.
\vspace{-2ex}
\subsection{First Match Algorithm }
The whole procedure of First Match is illustrated in Algorithm \ref{alg:fm}. Note that our target is to assign each worker to one question $s.t.$ the final question set $|Q_{T}|$ is minimized. In order to satisfy this target, the main idea of the First Match Algorithm is that whenever we consider a new worker, we simply look at the open questions and assign him to the First Matching open question without exceeding its remaining easiness score. If the $EI$ of this worker cannot fit into all open questions, at this time we open a new closed question. First match controls the number of open questions by adopting a lazy assignment, where it either assigns a worker to the First Matching open question or conservatively opens a new question when all open questions cannot hold the $EI$ of this worker.

We use the symbol $\rightarrow$ to denote the task assignment operation. In lines 1-3, $U_{i,j}$ denotes the remaining easiness score of question $q_{j}$ if answered by $w_{i}$. In lines 4-5, we open the first question by finding the smallest $U_{i',j'}$, then $w_{i'}$ is set to $w_{1}$ and the assignment $w_{1} \rightarrow q_{j'}$ is added to the assignment scheme $I$. In addition, $q_{j'}$ is added to the set of open questions $Q_{o}$. For the remaining workers $w_{i} \in W-\{w_{1}\}$ where $i>1$, we randomly index them and update their $U_{i,j'}$ (lines 8-9). For each worker $w_{i}$, the First Match algorithm then opens a question only if the current worker $w_{i}$ cannot fit in all previous open questions $Q_{o}$ and it selects a new open question that has the least remaining easiness score if answered by $w_{i}$ (lines 11-13). Otherwise, if $w_{i}$ can fit into more than one open questions in $Q_{o}$, he will be assigned to the question with the lowest index, that is, the question opened the earliest (lines 14-15). Finally, the assignment is added to $I$ and after the assignment, $U$ is updated (lines 16-17).

\noindent\textbf{The Time Complexity.}
We need $O(mn)$ steps to initialize $U_{i,j}$ and at most $O(mn)$ steps to figure out the min $U_{i',j'}$ respectively (line 1-4). In lines 10-15, for each worker $w_{i}$, we need $O(m)$ time to find either the min $U_{i,k}$ ($q_{k} \in Q - Q_{o}$, that is the set of closed questions) or the first open question $q_{k} \in Q_{o}$. And updating $U_{z,k} $, $\forall z > i$ requires additional $O(n)$ time cost in the worst case since the worker set $W$ of size $n$ is looped. Totally, the time complexity of the First Match Algorithm is $O(mn + mn + n(m+n))$, which is equal to $O(n^2 + nm)$.

\vspace{-2ex}
\subsection{Best Match Algorithm }
As mentioned above, the First Match Algorithm adopts a lazy question assignment scheme where a worker is always assigned to the first feasible question $q_{j}$ in the set of open questions $Q_{o}$. However, there exists a case that if $w_{i}$ is assigned to some other question $q_{k}$ in $Q_{o}$, the remaining easiness score of $q_{k}$ is less than that of $q_{j}$, namely, $0 \leq c_{k}-E^{2}I_{i,k} < c_{j}-E^{2}I_{i,j}$, and $k>j$. It is highly possible that we can yield more completed questions at the end of this batch if $w_{i} \rightarrow q_{k}$. However, First Match does not consider this situation.

In order to select the questions more carefully for the workers and further increase the number of potentially completed questions, we propose another approximation algorithm, namely Best Match Algorithm, as shown in Algorithm \ref{alg:bm}.

\begin{algorithm}[t]

        \caption{\small{Best Match Algorithm}}
        \label{alg:bm}
        \KwIn{$E^{2}I$, $W$, $Q$ and $C=c_{1}, c_{2},\ldots, c_{n}$}
        \KwOut{$Q_{T}$, question assignment scheme I}
        \For {$w_{i} \in W$}
        {
            \For {$q_{j} \in Q$    }
            {
                Let $U_{i,j} = c_{j} - E^{2}I_{i,j}$, the remaining $c_{j}$ after $w_{i} \rightarrow q_{j}$;
            }
        }
        $\forall w_{i} \in W, \forall q_{j} \in Q$, find $U_{i',j'}$ = min $U_{i,j}$;\\
        Let $w_{i'}$ be $w_{1}$, $I = \{w_{1} \rightarrow q_{j'}\}$; \\
        Let $Q_{o}$ denote the set of open questions;\\
        $Q_{o} = \{q_{j'}\}$;\\
        Randomly order the workers from $w_{2}$ to $w_{n}$;\\
        $\forall i > 1$, update $U_{i,j'} = U_{i,j'} - E^{2}I_{1,j'} $;\\
        \For {$w_{i} \in W-\{w_{1}\}$}
        {
            \eIf{$\forall q_{j} \in Q_{o}$, $U_{i,j} < 0 $}{
                Let $U_{i,k}$ = min $U_{i,j}$, $\forall q_{j} \in Q-Q_{o}$;\\
                Open the closed question $q_{k}$, set $Q_{o} = Q_{o} \cup \{q_{k}\}$; \\
                $I = I \cup \{w_{i} \rightarrow q_{k}\}$;\\
                $\forall z > i$, update $U_{z,k} = U_{z,k} - E^{2}I_{i,k} $;\\
            }{
                Find the best question $q_{z}$ $ s.t.$ $U_{i,z}$ = min $U_{i,j}$, $\forall q_{j} \in Q_{o}$;\\
                $I = I \cup \{w_{i} \rightarrow q_{z}\}$;\\
                $\forall v > i$, update $U_{v,z} = U_{v,z} - E^{2}I_{v,z} $;\\
            }
        }   
        	$Q_{T} = Q_{o}$ 
\end{algorithm}

The second algorithm, the Best Match algorithm, is similar to the First Match algorithm in lines 1-15. However, the main difference is that when $w_{i}$ can fit in more than one open question in $Q_{o}$, $w_{i}$ will be assigned to the question with the smallest remaining score incurred by him. This is revealed in lines 16-19, where the best question is first found by performing a traversal of the remaining easiness scores over the set of open questions $Q_{o}$ and then finding the minimum $U_{i,z}$. Therefore, worker $w_{i}$ is assigned to question $q_{z}$ in line 18. After that, all affected $U$ is updated in line 19.

\noindent\textbf{The Time Complexity.}
Similar with the First Match algorithm, we need $O(mn)$ steps to initialize $U_{i,j}$ and at most $O(mn)$ steps to figure out the min $U_{i',j'}$ respectively (lines 1-4). In lines 10-19, for each worker $w_{i}$, we need $O(m)$ time to find either the min $U_{i,k}$ ($q_{k} \in Q - Q_{o}$, which is the set of closed questions) (lines 12-14) or the min $U_{i,z}$ ($q_{z} \in Q_{o}$, which is the set of open questions) (lines 17-18). Furthermore, updating affected $U$ requires additional $O(n)$ time cost in the worst case since the worker set $W$ of size $n$ is looped. To sum up, the time complexity of the Best Match Algorithm is $O(mn + mn + n(m+n))$, which is equal to $O(n^2 + nm)$. The time complexities of the two algorithms are identical.

\subsection{The Approximation Ratios}
Next, we study the approximation ratios of the \textit{First Match Algorithm} and \textit{Best Match Algorithm} together since we prove that they have the same approximation ratios. We use \textit{XM} to represent both of them for the convenience of description.

Note that in Definition \ref{mcq}, our MCQ problem minimizes $|Q_{T}|$.

\begin{lemma} \label{lemma1}
Index the questions of $Q_{T}$ in the order opened by XM. Consider the $j$th ($j \geq 2$) question, any worker that was assigned to it cannot fit into questions opened prior to $q_{j}$.
\end{lemma}

\begin{proof}
    We show the proof of Lemma \ref{lemma1} by contradiction. Suppose there is a question $q_{k}$, where $q_{k} \in Q_{T}$ and $k < j$. Assume worker $w$ can be fit into $q_{k}$ and is assigned to $q_{j}$ for First Match Algorithm. However, if there is more than one question in $Q_{T}$ that can hold $w$, the First Match will pick the first satisfied question. Since $k < j$, $q_{k}$ should be selected to $w$, which forms a contradiction. Thus, Lemma \ref{lemma1} is proved under the First Match. The proof of Best Match algorithm is similar to the First Match and is omitted here.
\end{proof}

\begin{algorithm}[t]
        \caption{\small{The Helper Procedure}}
        \label{alg:HP}
        \For{$i = 1$ to $v-1$}
        {Let $X_{j'}$ be the nonempty set with the highest index;\\
            \eIf{$j' = i$}{ Stop.}
            {Let $w$ be the worker with smallest $E^{2}I$ to $q_{j'}$;\\
                Set $X_{i} \leftarrow X_{i} \cup \{w\}$;
                and $X_{j'} \leftarrow X_{j'} \setminus \{w\}$;}
        }
\end{algorithm}

In order to continue finding the approximation ratios, we divide workers and questions of $Q_{T}$ into two types respectively as:
\begin{itemize}[leftmargin=*]
\item{\textbf{Experts/normal workers.}} A worker $w_{i}$ is called an expert if his $E^{2}I$ is greater than the half remaining easiness score of all open questions, which means $\forall q_{j} \in Q_{T}$, $E^{2}I_{ij} \geq \frac{c_{j}}{2}$, where $Q_{T}$ refers to the set of open questions. Otherwise, he is a normal worker.
\item{\textbf{Type N/E questions.}} A question is defined to be Type N if it is only assigned to normal workers and is of Type E if it is not Type N, namely, it has at least one assigned expert.
\end{itemize}

We know that XM can ultimately get $|Q_{T}|$ questions assigned. For a given integer $v$, $2 \leq v \leq |Q_{T}|$, select $v$ questions from $Q_{T}$ and index them in the order that opened by XM. Let $X_{j}$ be the set of workers assigned by XM to the $j$th question, $j=1, 2, \dots, v$. 

Then we manually partition $Q_{T}$ produced by XM into two sets. The first set includes only Type N questions, and the second set includes the remaining questions produced by XM. Assume there are $c$ experts among the workers. Since every worker is assigned to one question, then $|Q_{T}|-c$ is the number of Type N questions and $c$ is the number of Type E questions. Index the questions in the first set in the order they are opened, from $1$ to $|Q_{T}|-c$.

Based on Lemma \ref{lemma1}, we construct a lower bound on the optimal $|Q_{T}|$ with the help of Algorithm \ref{alg:HP}: the Helper Procedure. In the Helper Procedure, workers assigned to latter questions in $Q_{T}$ are picked and re-assigned to former questions (line 6-7). However, according to Lemma \ref{lemma1}, workers are not able to fit into previous questions in First Match or Best Match. Therefore, every assignment in the Helper Procedure will cause the overused situation of some question $q_{j}$, namely, the joint $E^{2}I$ of workers assigned to $q_{j}$ is larger than its remaining easiness score $c_{j}$.

So we let $v = |Q_{T}|-c$ and apply the Helper Procedure to the set of Type N questions. Assume the Helper Procedure can produce $m$ questions, there are at least $m-1$ questions are overused. This can be easily derived. If the number of overused questions is less than $m-1$, say, equal to $k$, the Helper Procedure tends to pick another worker who was assigned to the $m^{th}$ question and re-assign him to $(k+1)^{th}$ question which will be overused. And this process will continue until no former questions before the $m^{th}$ question are available, which means the number of overused questions is at least $m-1$.

With the help of the Helper Procedure, Lemma \ref{lemma2} is proposed as follows:

\begin{lemma} \label{lemma2}
    $\frac{c}{2} + m^{*} - 1 < |Q^{*}_{T}|$, where $Q^{*}_{T}$ is the optimal solution and $m$ is produced by the Helper Procedure on the Type N questions of $Q^{*}_{T}$.
\end{lemma}

\begin{proof}
    Assume the optimal solution of the MCQ problem finally achieves $Q_{T}^{*}$ as the final set of assigned questions. $Q_{T}^{*}$ is divided into two sets: the first set contains the Type N questions while the second set contains the Type E question. According to the above, the size of Type E questions is exactly $c$. And assume there are $y$ questions of Type N. We apply the Helper Procedure to the first set and obtain $m^{*}$ questions with at least $m^{*}-1$ questions are overused, where $m^{*}-1 < y$. Furthermore, from the definition of \textit{experts}, we know that there exist no two experts that can be assigned to one identical question. Therefore, since $|Q^{*}_{T}| = c + y$, and $m^{*}-1 < y$, we then yield $|Q^{*}_{T}| > c+ m^{*}-1 > \frac{c}{2}+ m^{*}-1$. 
\end{proof}

\begin{theorem}
    XM Algorithms have approximation ratios of $2 + \frac{2}{|Q^{*}_{T}|}$.
\end{theorem}
\begin{proof}
    As mentioned previously, the Helper Procedure produces $m^{*}$ questions with at least $m^{*}-1$ overused, which means at least $m^{*}-1$ questions contain 2 workers. Therefore, we have $m^{*} \geq \frac{|Q_{T}|-c}{2}$ or $|Q_{T}| - m^{*} - c \leq m^{*}$. Combining Lemma \ref{lemma2}, we have the following:\vspace{-2ex}
    \begin{equation}
    \begin{split}
    &|Q_{T}| - m^{*} - c  + (\frac{c}{2}-1)\leq m^{*} + (\frac{c}{2}-1) < |Q^{*}_{T}|\\
    \Leftrightarrow \quad &|Q_{T}| < |Q^{*}_{T}| + m^{*} + \frac{c}{2} + 1 < |Q^{*}_{T}| + (m^{*} + \frac{c}{2} -1) + 2\\
    \Leftrightarrow \quad &|Q_{T}| < 2|Q^{*}_{T}| + 2\\
    \Leftrightarrow \quad &\frac{|Q_{T}|}{|Q^{*}_{T}|} < 2 + \frac{2}{|Q^{*}_{T}|},
    \end{split}\notag
    \end{equation}
where the first line is obtained by adding $\frac{c}{2}-1$ to both sides of formula $|Q_{T}| - m^{*} - c \leq m^{*}$. 
\end{proof}
\vspace{-2ex}

%% file: experiment_v1.tex
\section{Experiments} \label{experiment}

We extensively evaluate our MCQ approaches on both synthetic and real datasets. The results are presented and analyzed in this section. Before presenting the results, we first introduce the datasets we use.
\begin{figure*}[h] \centering
	\subfigure[Accuracy] { \label{fig:syn_m_accuracy_0.2}
		\includegraphics[height = 1.2in,width=0.37\columnwidth]{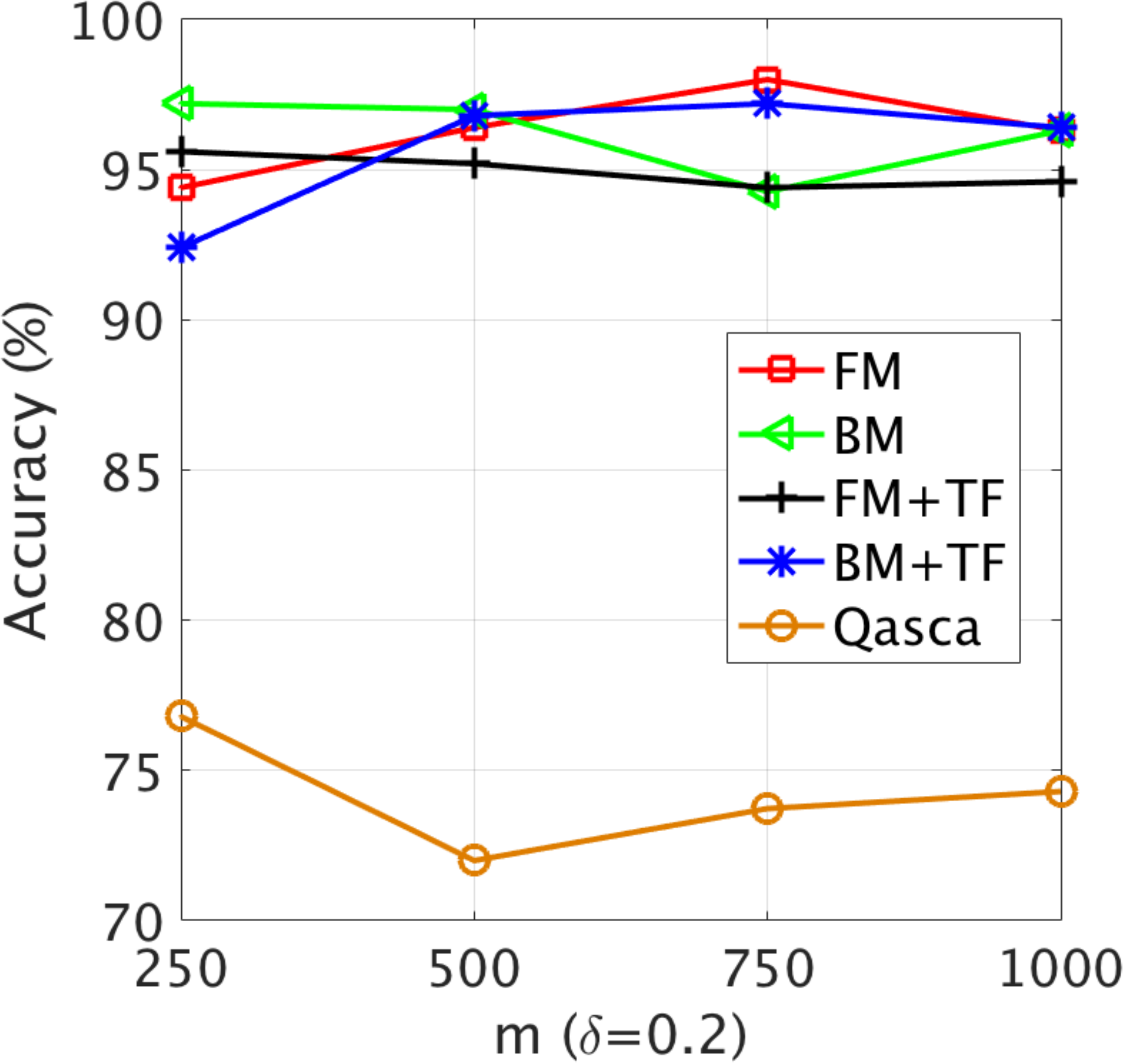}
	}	
	\subfigure[Accuracy] { \label{fig:syn_m_accuracy_0.3}
		\includegraphics[height = 1.2in,width=0.37\columnwidth]{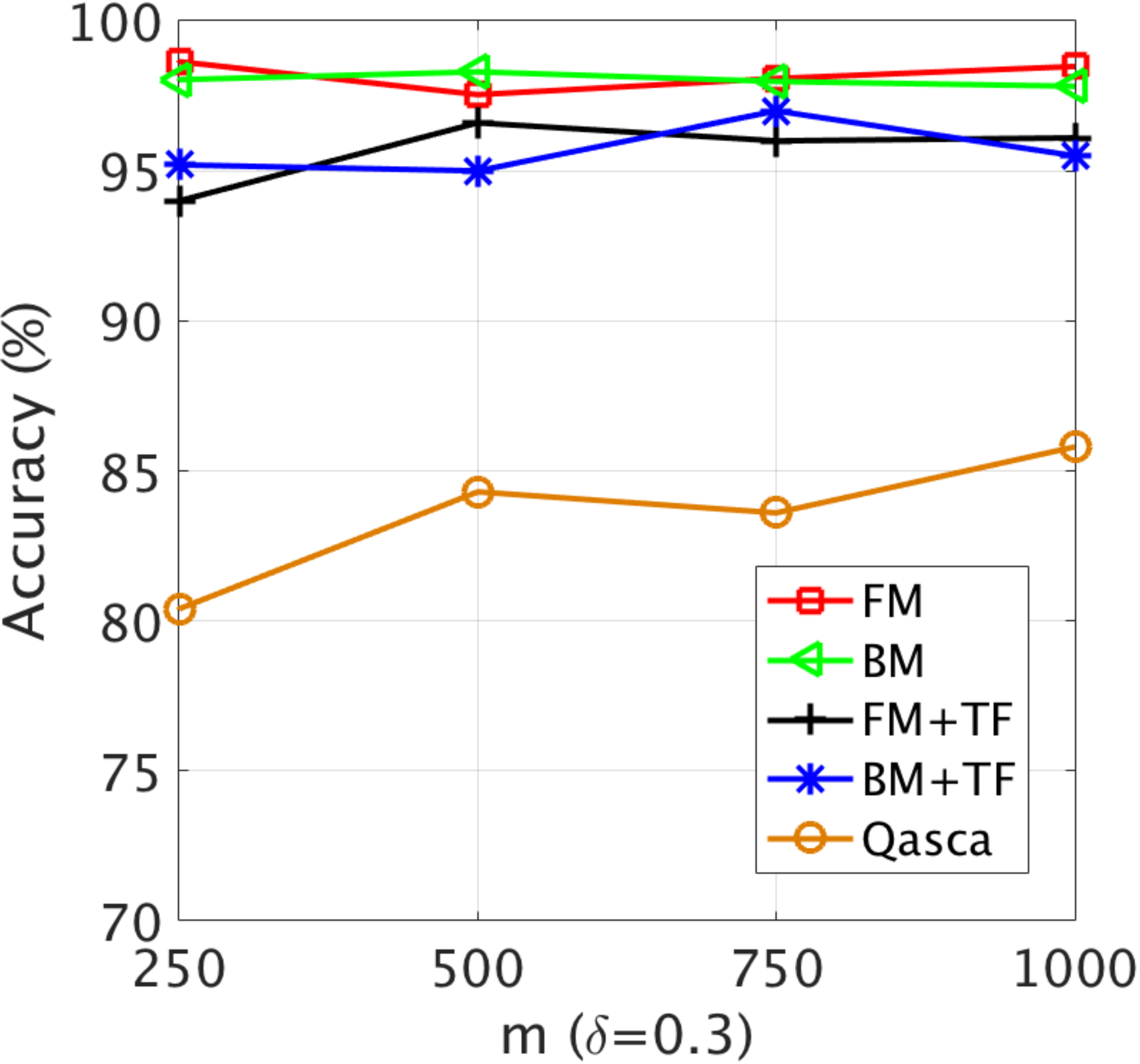}
	}
	\subfigure[Accuracy] { \label{fig:syn_m_accuracy_0.4}
	\includegraphics[height = 1.2in,width=0.37\columnwidth]{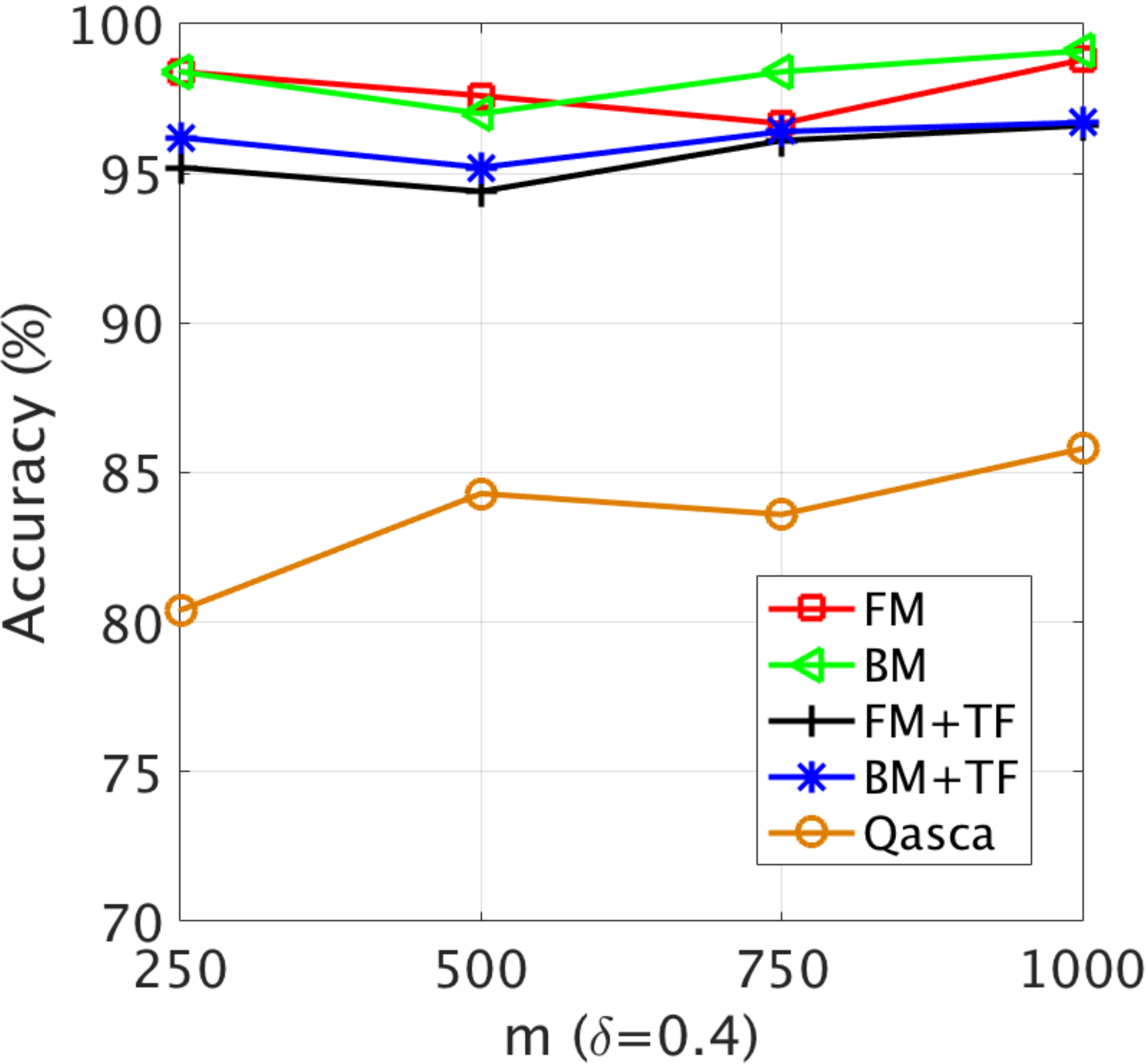}
	}
	\subfigure[Accuracy] { \label{fig:syn_m_accuracy_0.5}
	\includegraphics[height = 1.2in,width=0.37\columnwidth]{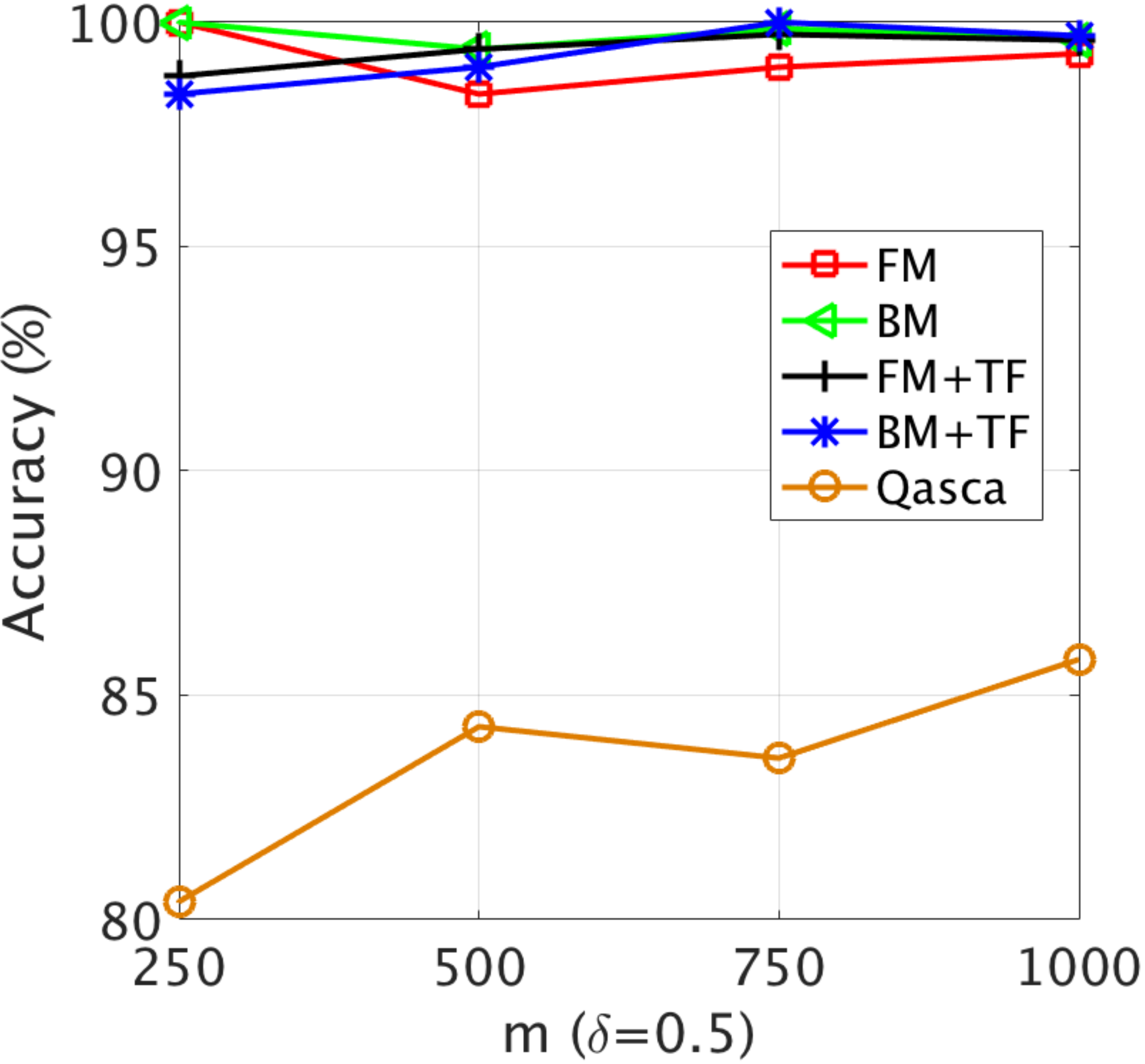}
	}
	\subfigure[Accuracy] { \label{fig:syn_m_accuracy_0.6}
	\includegraphics[height = 1.2in,width=0.37\columnwidth]{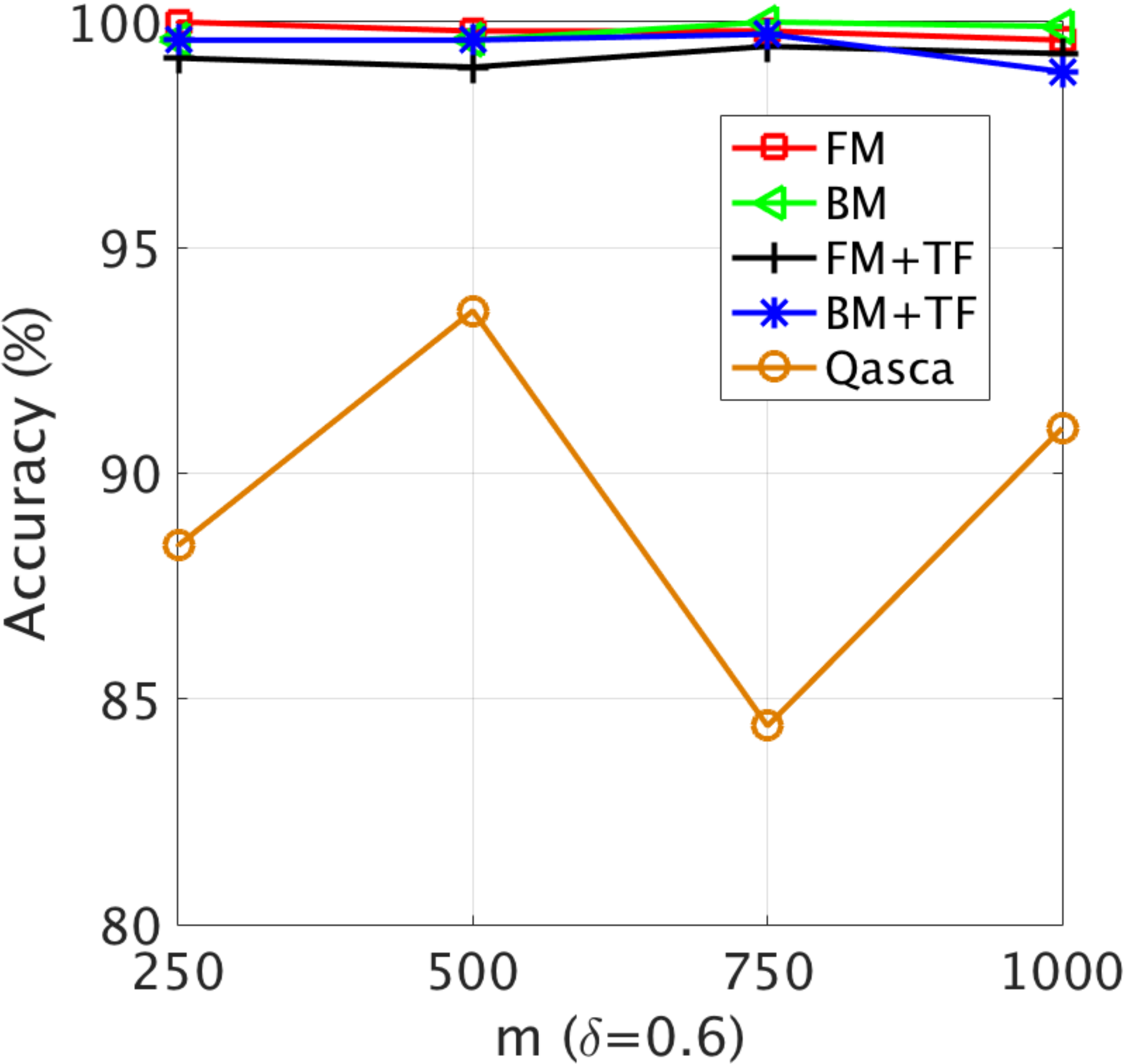}
	}

\subfigure[$\#$ Required Batches] { \label{fig:syn_m_batches_0.2}
	\includegraphics[height = 1.2in,width=0.37\columnwidth]{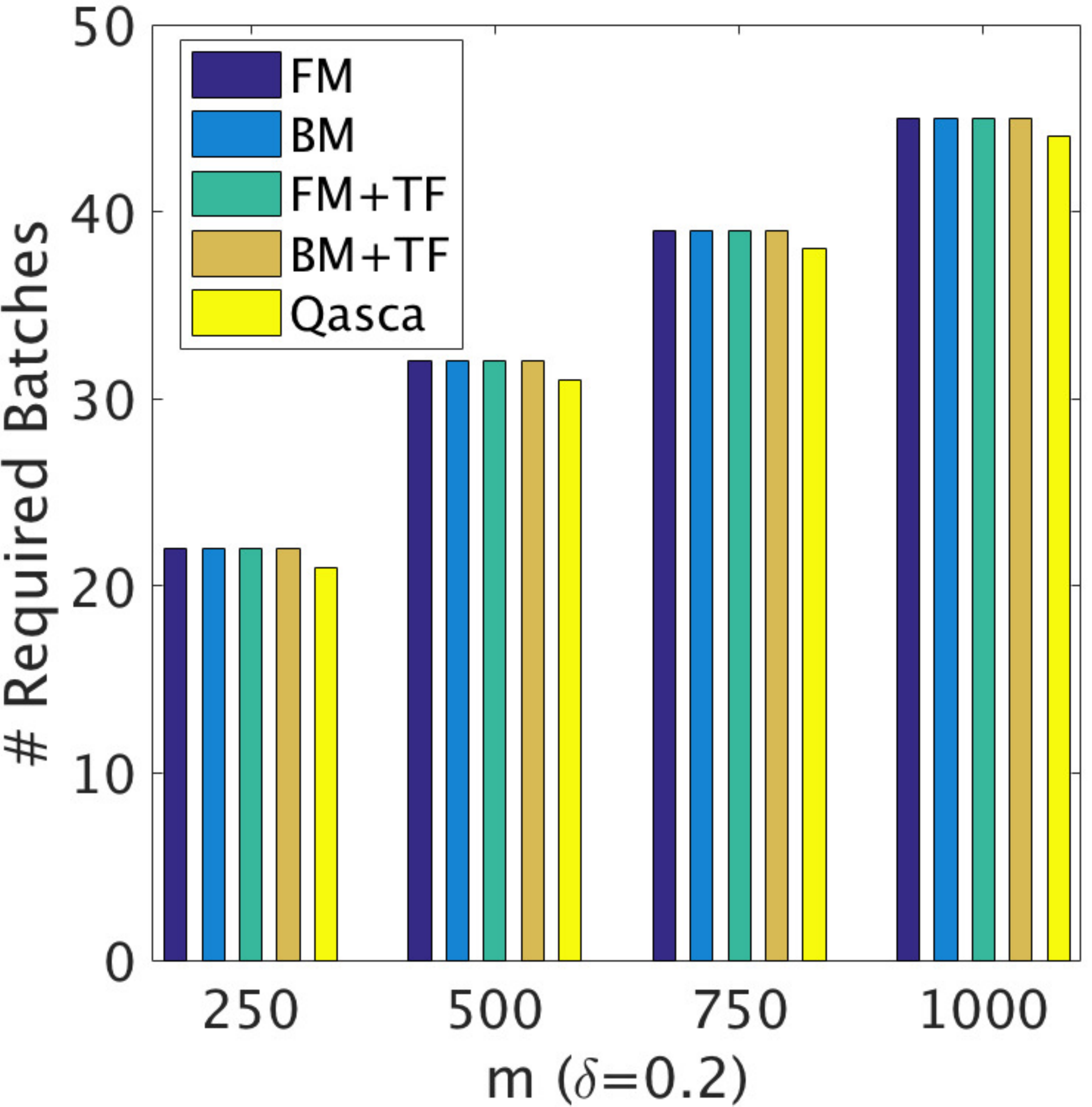}
}	
\subfigure[$\#$ Required Batches] { \label{fig:syn_m_batches_0.3}
	\includegraphics[height = 1.2in,width=0.37\columnwidth]{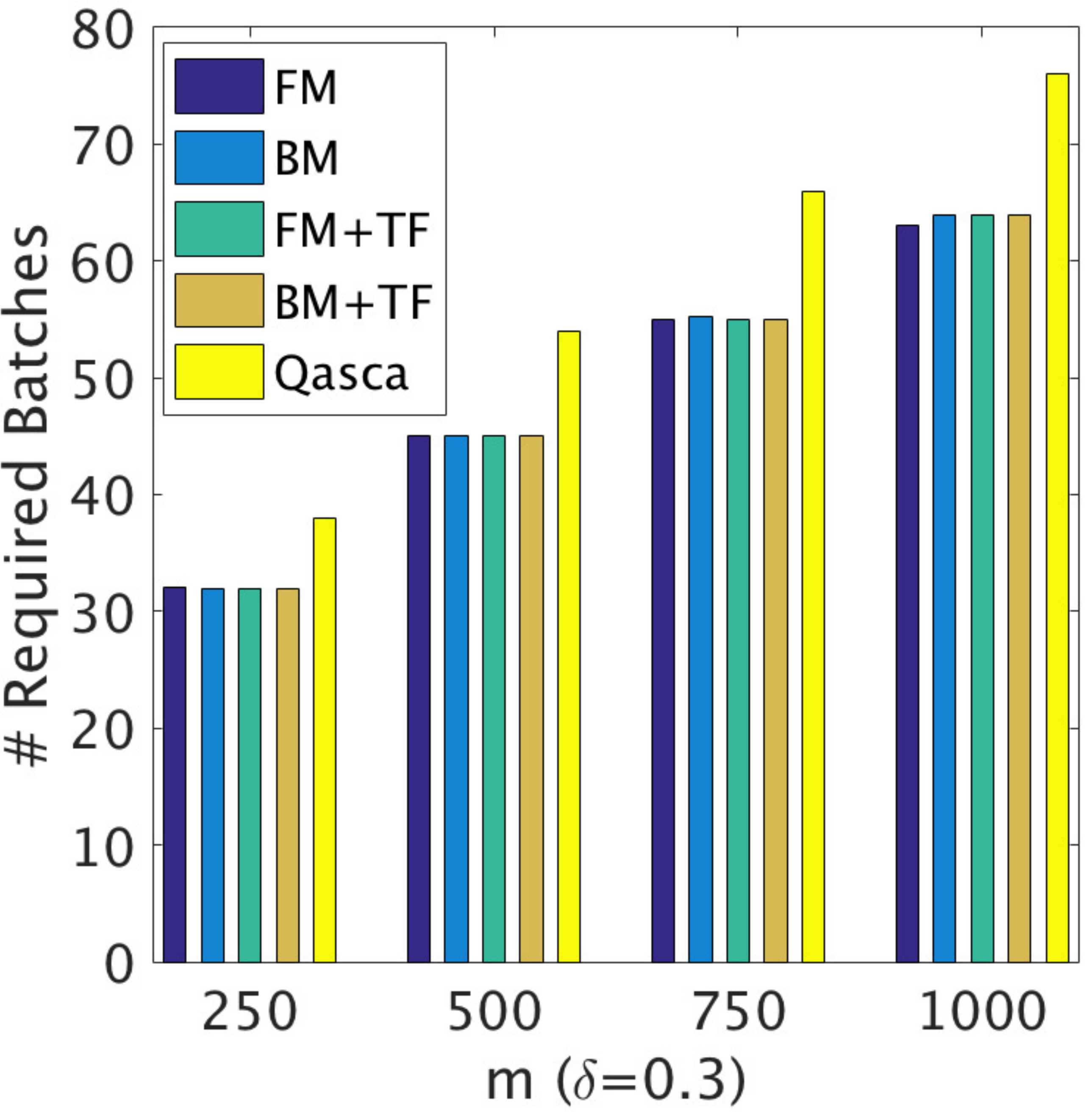}
}
\subfigure[$\#$ Required Batches] { \label{fig:syn_m_batches_0.4}
	\includegraphics[height = 1.2in,width=0.37\columnwidth]{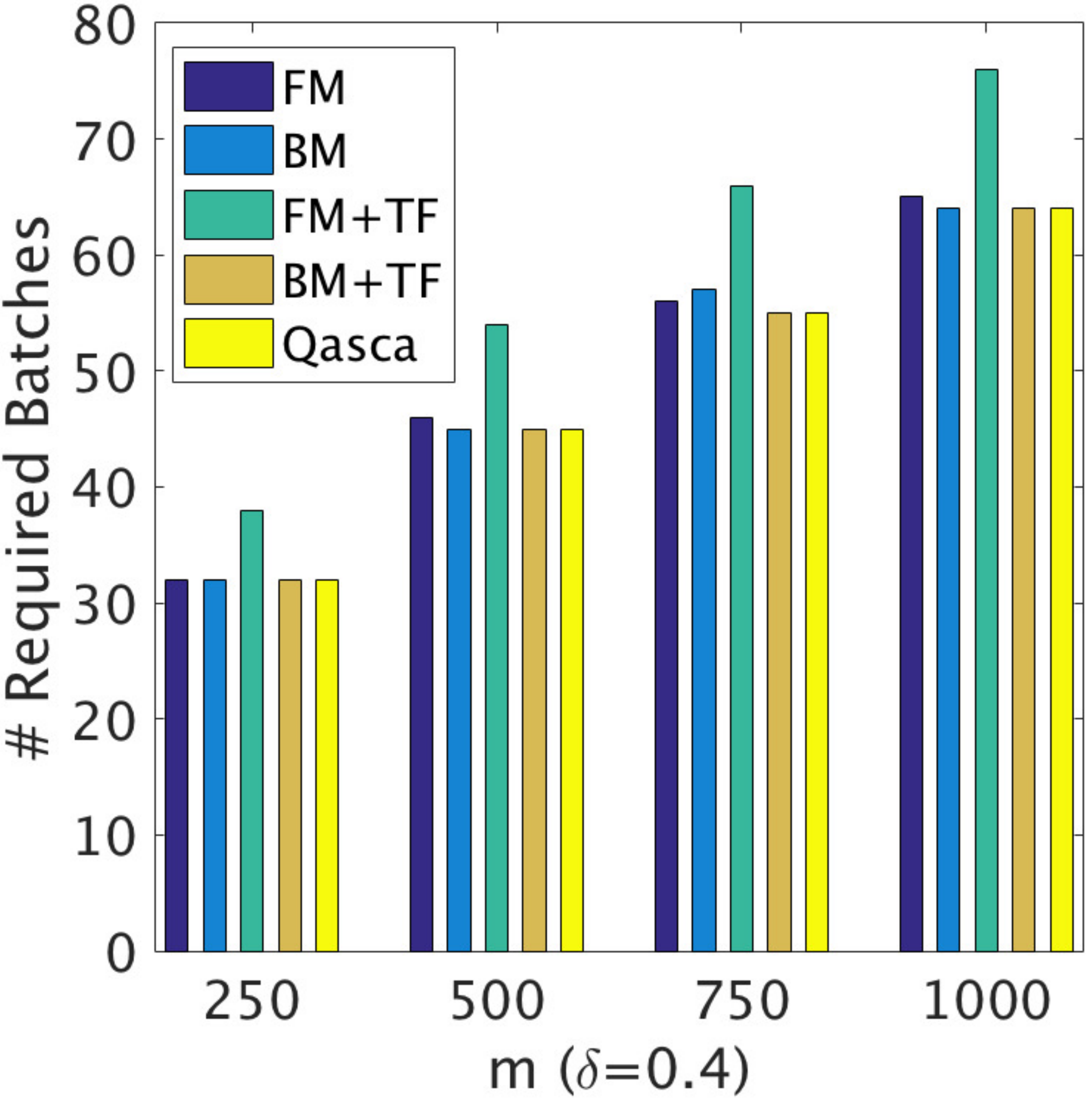}
}
\subfigure[$\#$ Required Batches] { \label{fig:syn_m_batches_0.5}
	\includegraphics[height = 1.2in,width=0.37\columnwidth]{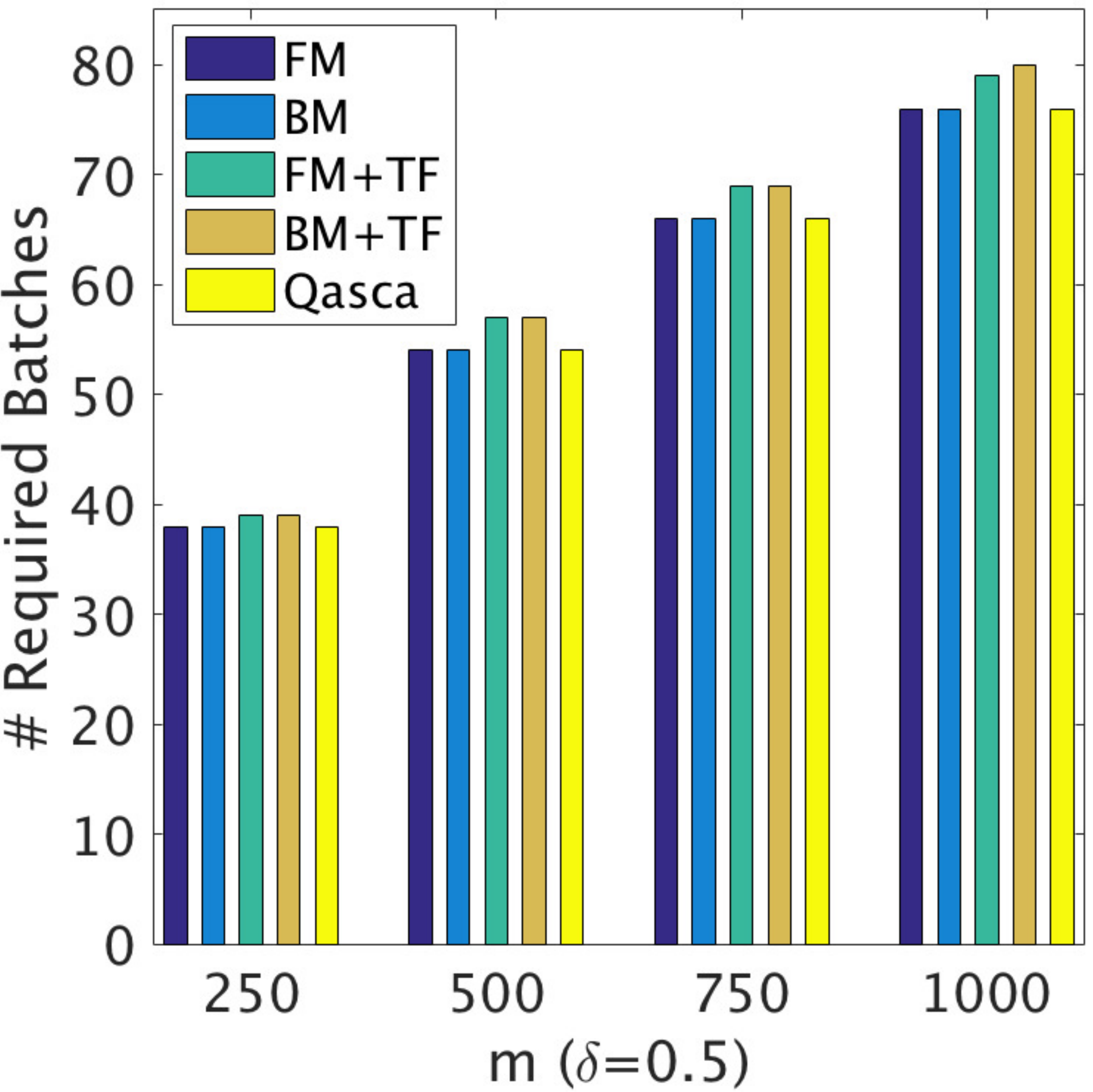}
}
\subfigure[$\#$ Required Batches] { \label{fig:syn_m_batches_0.6}
	\includegraphics[height = 1.2in,width=0.37\columnwidth]{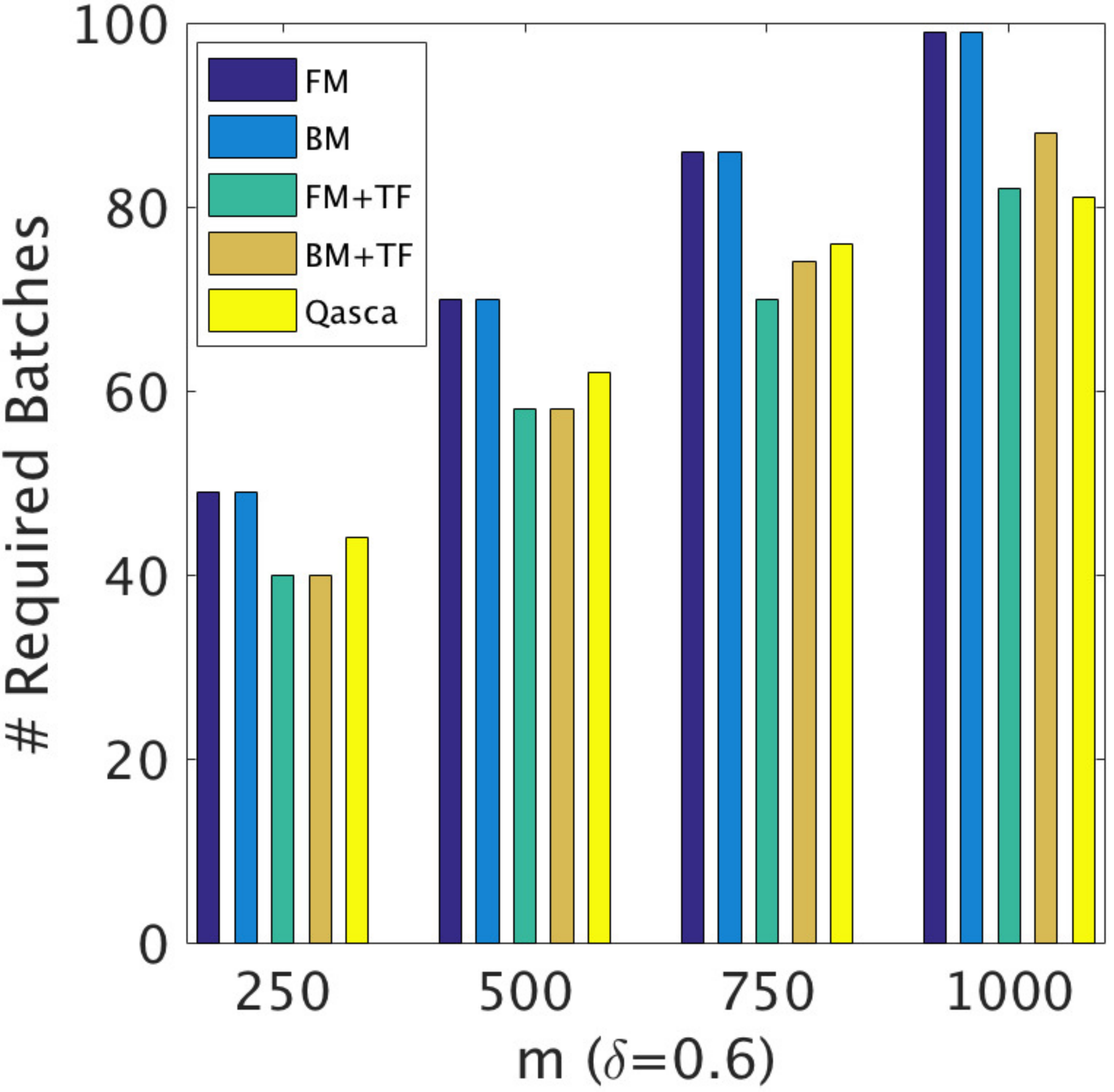}
}\vspace{-2ex}
	\caption{ The Effect of $m$ on Synthetic Data}\vspace{-3ex}
	\label{fig:syn_m}
\end{figure*}

\vspace{-2ex}
\subsection{Data Sets}
We use both real and synthetic data to study our proposed MCQ approaches. Specifically, for real data, we use one of the CrowdFlower datasets \cite{CrowdFlower} which contains worker arrival times and the Duck Identification dataset from \cite{TruthInference} that contains worker labels.

\noindent\textbf{CrowdFlower Dataset.} The dataset we use is called \textit{Relevance of terms to disaster relief topics}, where workers are asked to label the relevance of pairwise terms related to disaster relief. Most importantly, the dataset contains the concrete dates and times of all worker submissions. In total, there are 1,400 workers and 18,062 worker submissions. On the average, each worker has 12.9 submissions. The time span is 20 days.

\noindent\textbf{Duck Identification Dataset.} This dataset is released by the project in \cite{Zheng:2017:TIC:3055540.3055547}. It contains 108 questions and 39 workers, and every worker provides answers to all questions, where workers are asked to label 0 or 1 to denote whether the task contains a duck. All the questions are gold standard questions with true answers.

\noindent\textbf{Synthetic Dataset.} For Synthetic dataset, we randomly generate the truths of \textit{binary-choice questions}. The worker quality and question easiness are drawn from Gaussian distributions, which are simply utilized to simulate the worker selection probability. Note that our parameter estimation model does not initialize the parameters to the above distributions.

\begin{figure}[htbp]
	\centering
	\hfill
	\subfigure[Accuracy]{\includegraphics[height = 1.35in,width=0.48\linewidth]{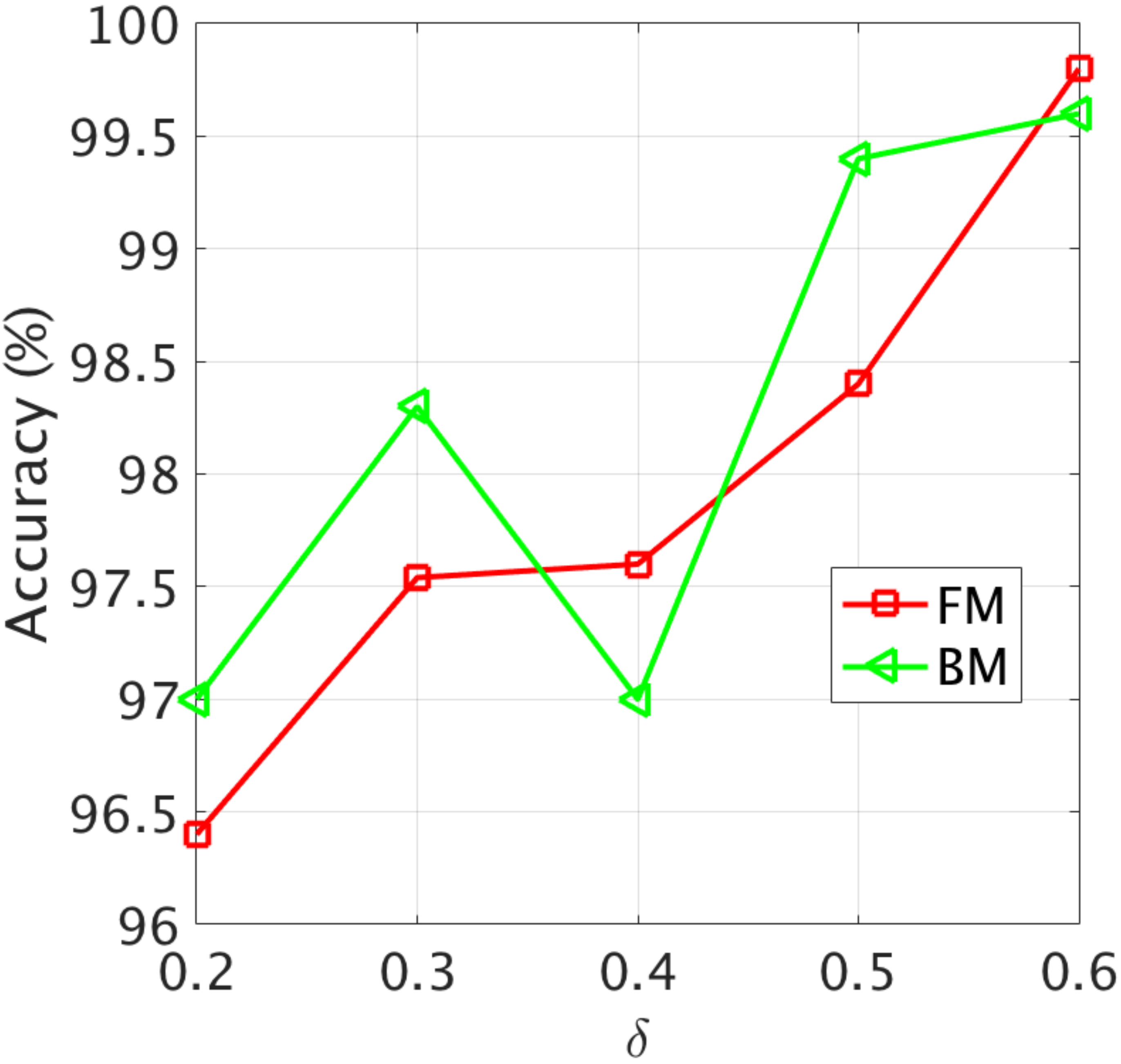}
		\label{fig:syn_delta_accuracy}}
	\hfill
	\subfigure[$\#$ Required Batches]{\includegraphics[height = 1.35in,width=0.48\linewidth]{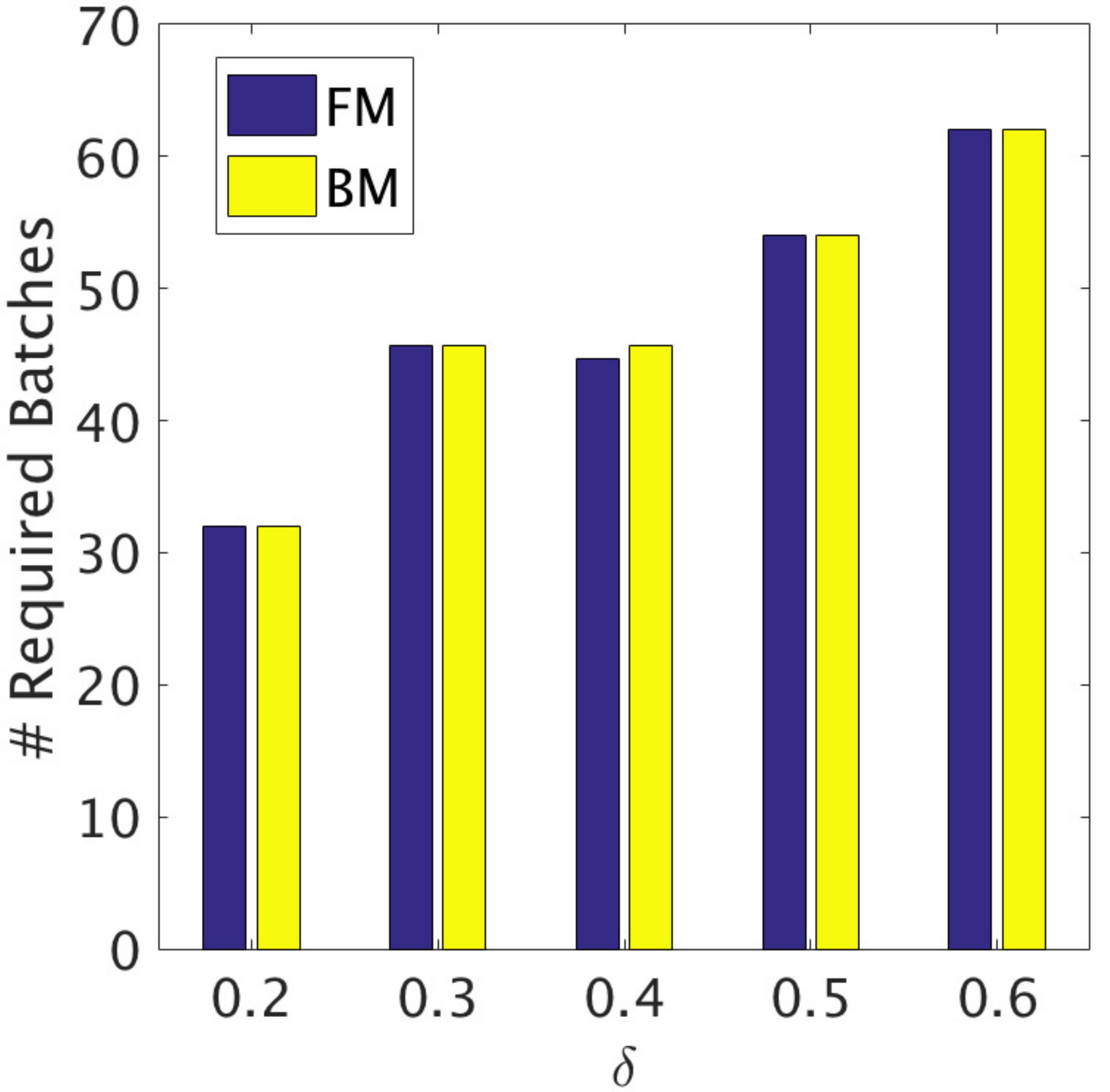}
		\label{fig:syn_delta_batches}}
	\hfill \vspace{-2ex}
	\caption{The Effect of $\delta$ on Synthetic Data}\vspace{-3ex}
	\label{fig:syn_delta_b}	
\end{figure}

\begin{table}[t!]    
	\centering

		\caption{Parameter Settings for Synthetic Experiment}
		\bigskip
		\label{syn_paras}{\small\vspace{-3ex}
			\begin{tabular}{c|c|c}
				
				Parameter & Description & Value\\
				\hline\hline
				$m$ & $\#$ questions & 250, \textbf{500}, 750, 1,000\\
				
				$\delta$ & easiness score threshold & 0.2, \textbf{0.3}, 0.4, 0.5, 0.6\\
				
				$rep$ & $\#$ repetitions in qasca & 1, \textbf{3}, 4\\
				
				$k$ & top-$k$ tasks in qasca & 1\\
				
				$\lambda$ & the worker arriving rate & 1 \\    
				\hline
			\end{tabular}}\vspace{-3ex}

	\end{table}

\subsection{Baseline and Evaluation Metrics}
Our \textit{synchronized task assignment} scheme executes in batches. At the end of each batch, when all worker answers are recorded, the iterative parameter estimation model starts to run, parameters like worker expertise, question easiness, and answer confidence are adaptively inferred and leveraged to perform task assignment in next batch. Therefore, the synchronized task assignment requires no historical performances or profiles of workers. With the above assumption, we adopt the online top-$k$ task assignment policy (using F-score) as the baseline (denoted as Qasca) since the way we achieve and update the worker quality is similar \cite{zheng2015qasca}. Furthermore, the online scheme is enforced to run in batches in order to compare with ours. We use FM and BM to denote First Match Greedy and Best Match Greedy respectively.

We compare our solutions with Qasca in two evaluation metrics: 1) the accuracy of the returned results; 2) the number of required batches to complete all questions. For all solutions, we repeat the question assignment for 100 times and record the average values of each evaluation metric.

\subsection{Experiments on Synthetic Data}
\subsubsection{Experimental Settings}
According to the synthetic dataset mentioned above, the worker expertise $e$ and question easiness $d$ are drawn from Gaussian distributions. Specifically, $e \sim \mathcal{N}(0.7, 0.1)$ and $d \sim \mathcal{N}(0.9, 0.03)$. With such distributions, the worker expertise is mostly set to be larger than 0.5 since we think the worst worker is the one who chooses randomly and the largest value of worker expertise is set to be 1. We adjust $d$ to a relatively high value to guarantee that the correct answers are more than erroneous ones. Note that we control the value of $d$ without exceeding 1 due to its definition. With the above settings, a worker is simulated to select a correct answer with the probability $p = \frac{1}{1+\mathrm{e}^{-e \cdot d}}$ \cite{whitehill2009whose}. For the parameter estimation model, we initialize the worker expertise, question easiness and answer confidence to 0.5 for every worker, question, and answer respectively. In order to accelerate the inference, we consider the influences between answers \cite{yin2008truth}. In binary-choice questions, the high confidence score of one answer indicates the low confidence score of the other. Therefore, we update their confidence scores by using $sc.c(a_{1})' = sc.c(a_{1}) - sc.c(a_{2})$ and $sc.c(a_{2})' = sc.c(a_{2}) - sc.c(a_{1})$, where $a_{1}$ and $a_{2}$ are the two answers for each question. In addition, we replace our parameter estimation model in both FM and BM with TRUTHFINDER (TF) \cite{yin2008truth} to evaluate the improvement of our proposed model, which are represented as FM+TF and BM+TF, respectively. In FM+TF and BM+TF, the parameter estimation model TF simply contains the inference between worker expertise and answer confidence (the left cycle), the inference between answer confidence and question easiness is removed and treated as an external unit that is calculated for one iteration after TF reaches a steady state. The above settings of FM+TF and BM+TF are identical with those in FM and BM. The settings of other important parameters are revealed in Table \ref{syn_paras}. For FM, BM, FM+TF and BM+TF, we vary $\delta$ from 0.2 to 0.6 ($\delta=0.3$ means that the answer confidences have to reach 0.65 and 0.35). Answers with higher confidence will be returned as the truths. For Qasca, the number of repetitions of every question has to be fixed \cite{zheng2015qasca}. And we set $rep=1$ for $\delta=0.2$, $rep=3$ for $\delta=0.3, 0.4, 0.5$ and $rep=4$ for $\delta=0.6$. The reason is that when we set $\delta=0.3$ (for example), a question in FM, BM, FM+TF, and BM+TF is assigned to 3 workers on average. Therefore, $rep=3$ makes sure that the total number of repetitions (or assignment) is identical for all solutions, thus the comparisons of the overall batches and accuracy are fair (the same reason for other cases of $\delta$ and $rep$). In addition, $k=1$ ensures the case that a worker in Qasca is assigned to 1 question, which is the same as FM, BM, FM+TF, and BM+TF. The truths of Qasca are inferred using EM algorithm. For the above five methods, the worker processing time is set to one batch per question, and the worker arriving rate ($\lambda$) is set to one worker per batch.

\subsubsection{Experimental Results} \hfill 
\begin{figure}[t!]\centering
	\scalebox{0.15}[0.15]{\includegraphics{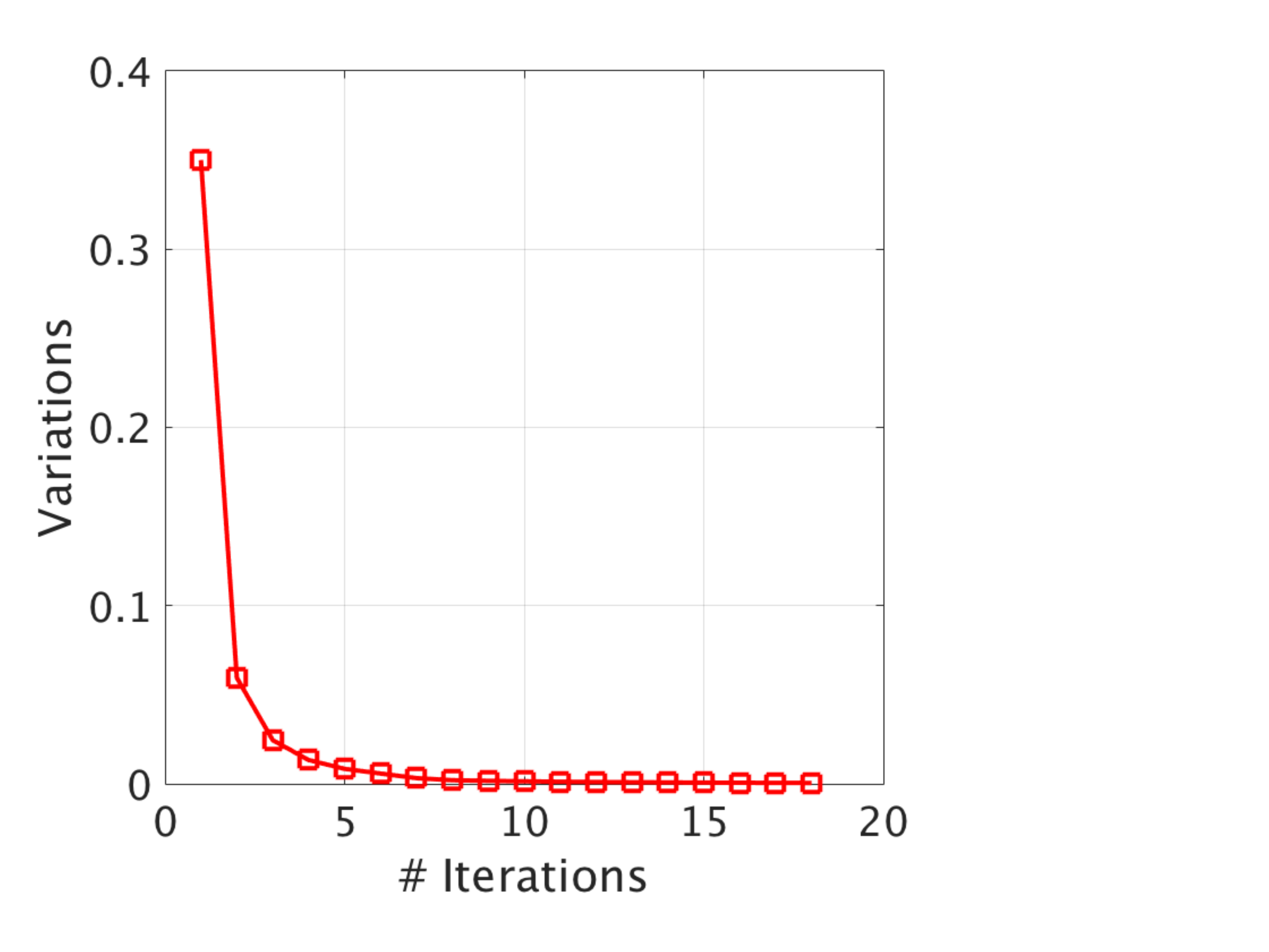}}\vspace{-2ex}
	\caption{\small Variations of Answer Confidence.}\vspace{-3ex}
	\label{fig:r_conv}
\end{figure}
\begin{figure*}[h] \centering
	\subfigure[Accuracy] { \label{fig:r_m_accuracy}
		\includegraphics[height = 1.35in,width=0.45\columnwidth]{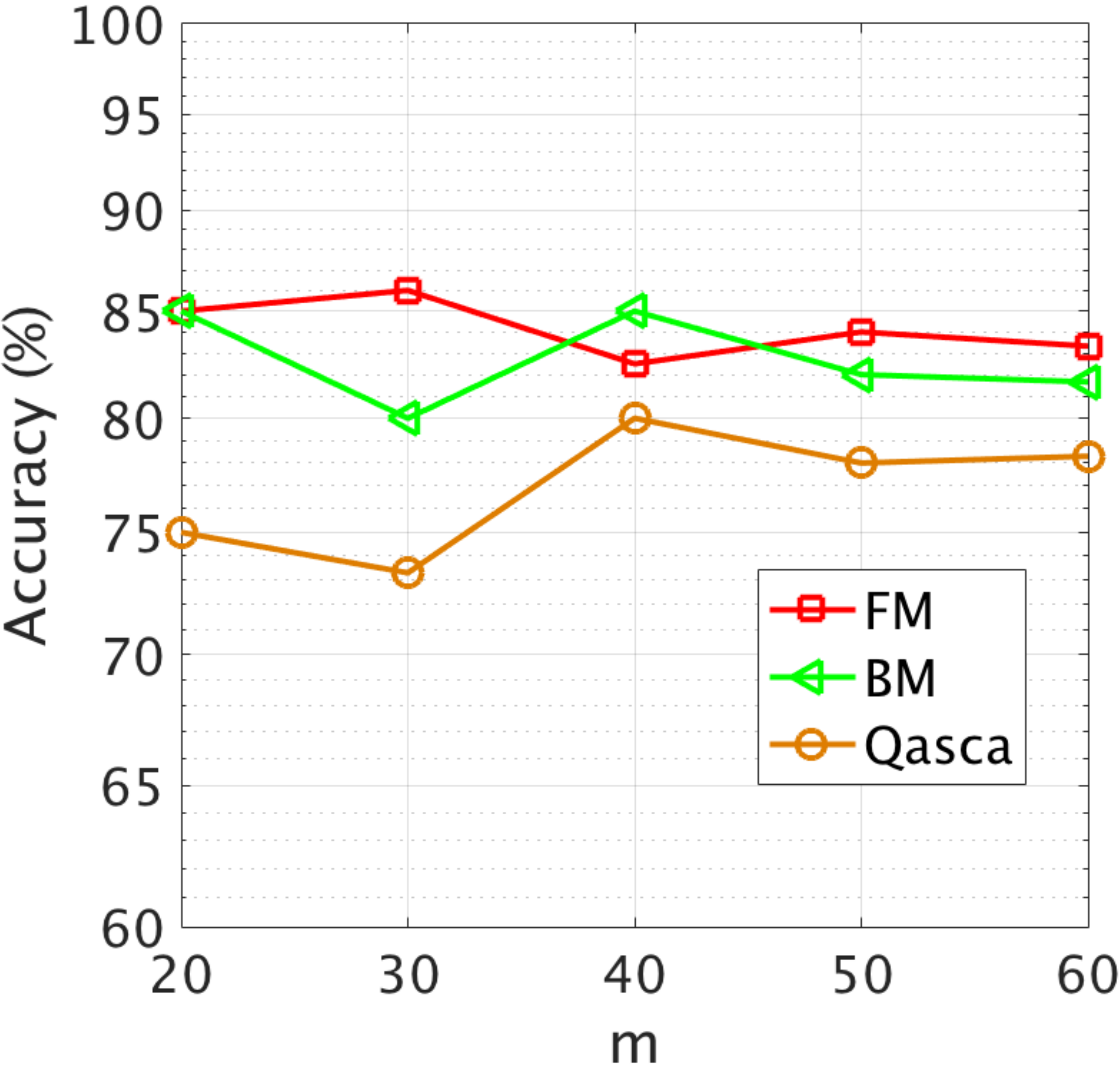}
	}
	\subfigure[$\#$ Required Batches] { \label{fig:r_m_batches}
		\includegraphics[height = 1.35in,width=0.45\columnwidth]{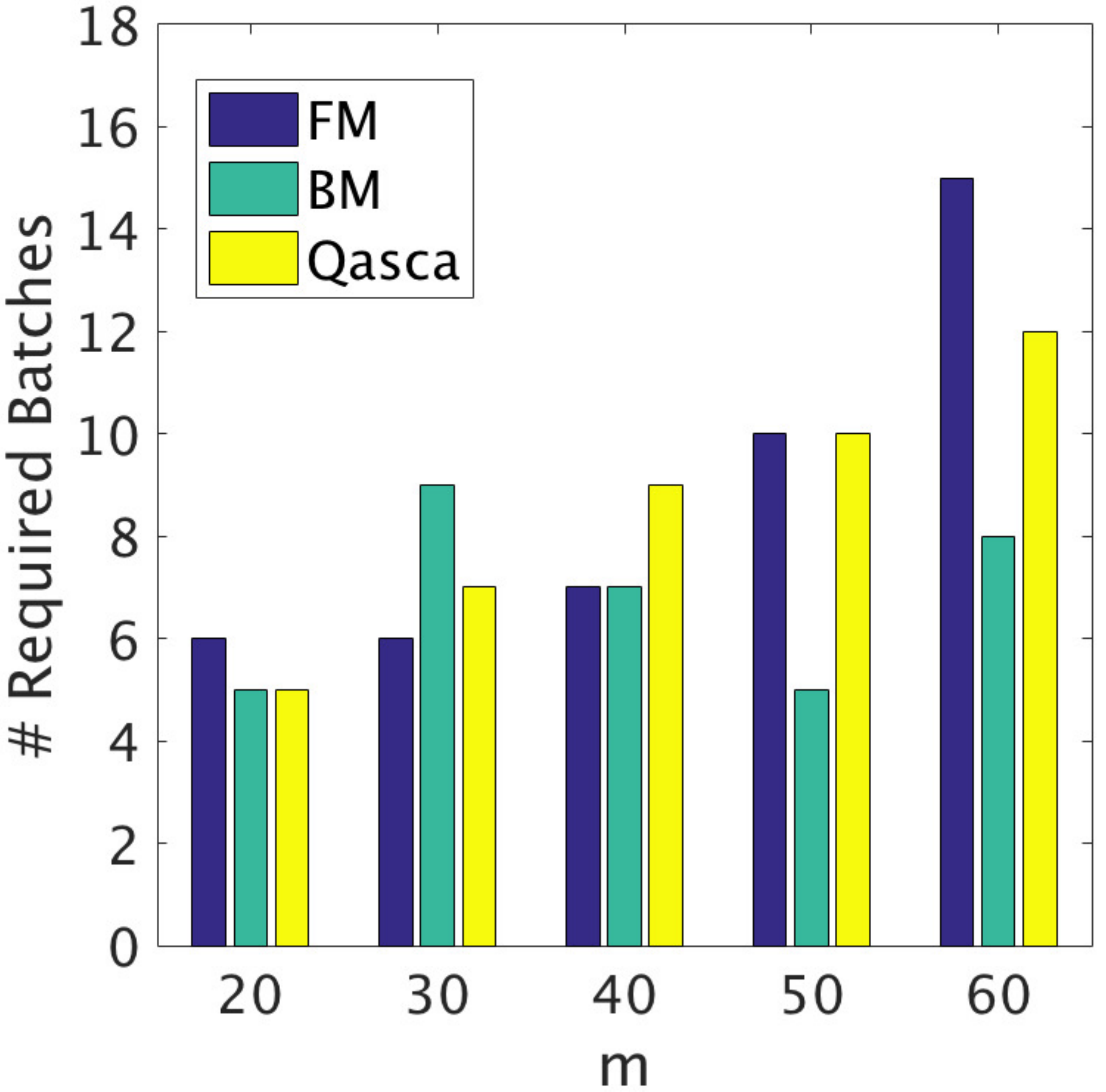}
	}
	\subfigure[Accuracy] { \label{fig:r_delta_accuracy}
		\includegraphics[height = 1.35in,width=0.45\columnwidth]{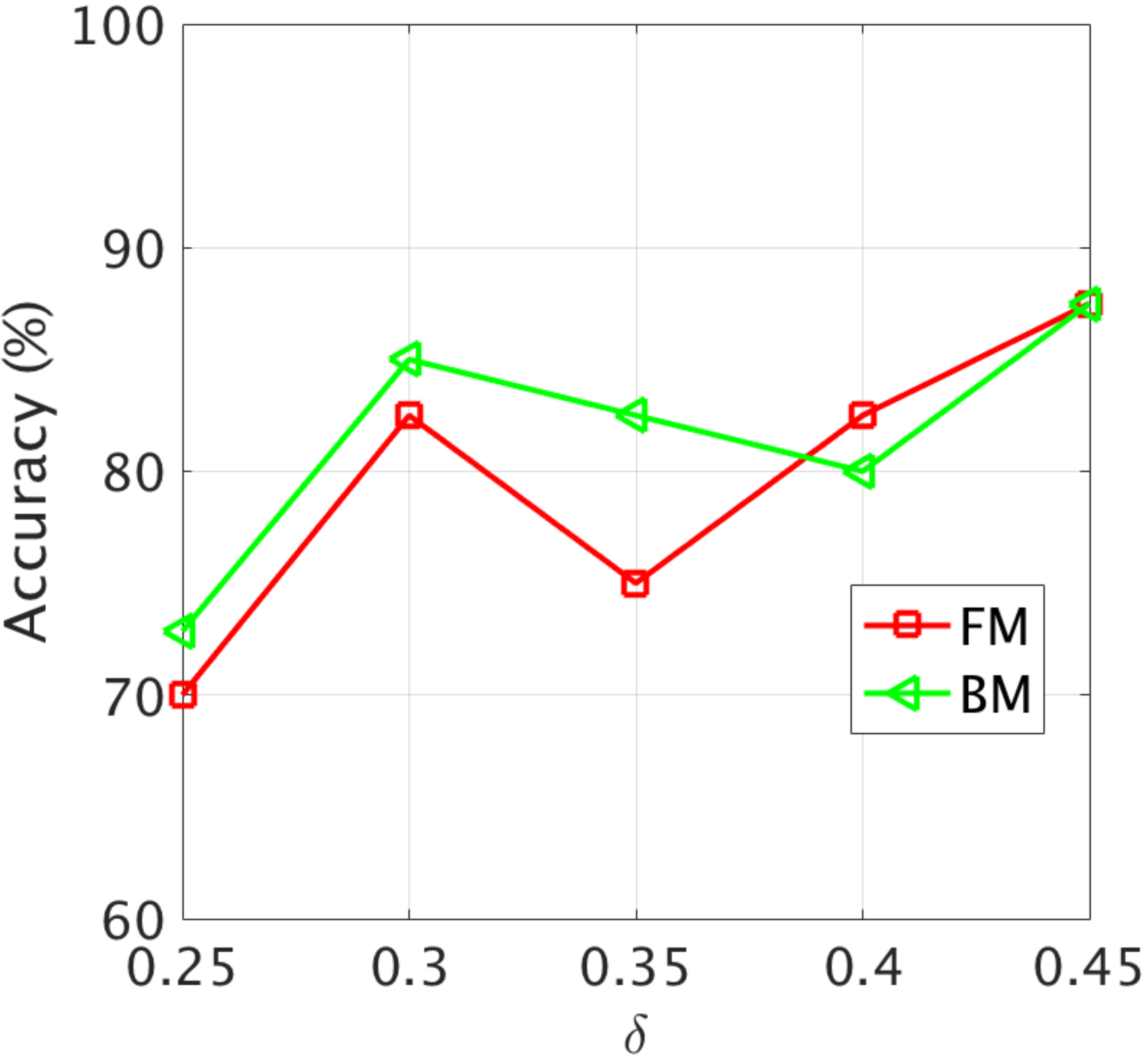}
	}
	\subfigure[$\#$ Required Batches] { \label{fig:r_delta_batches}
		\includegraphics[height = 1.35in,width=0.45\columnwidth]{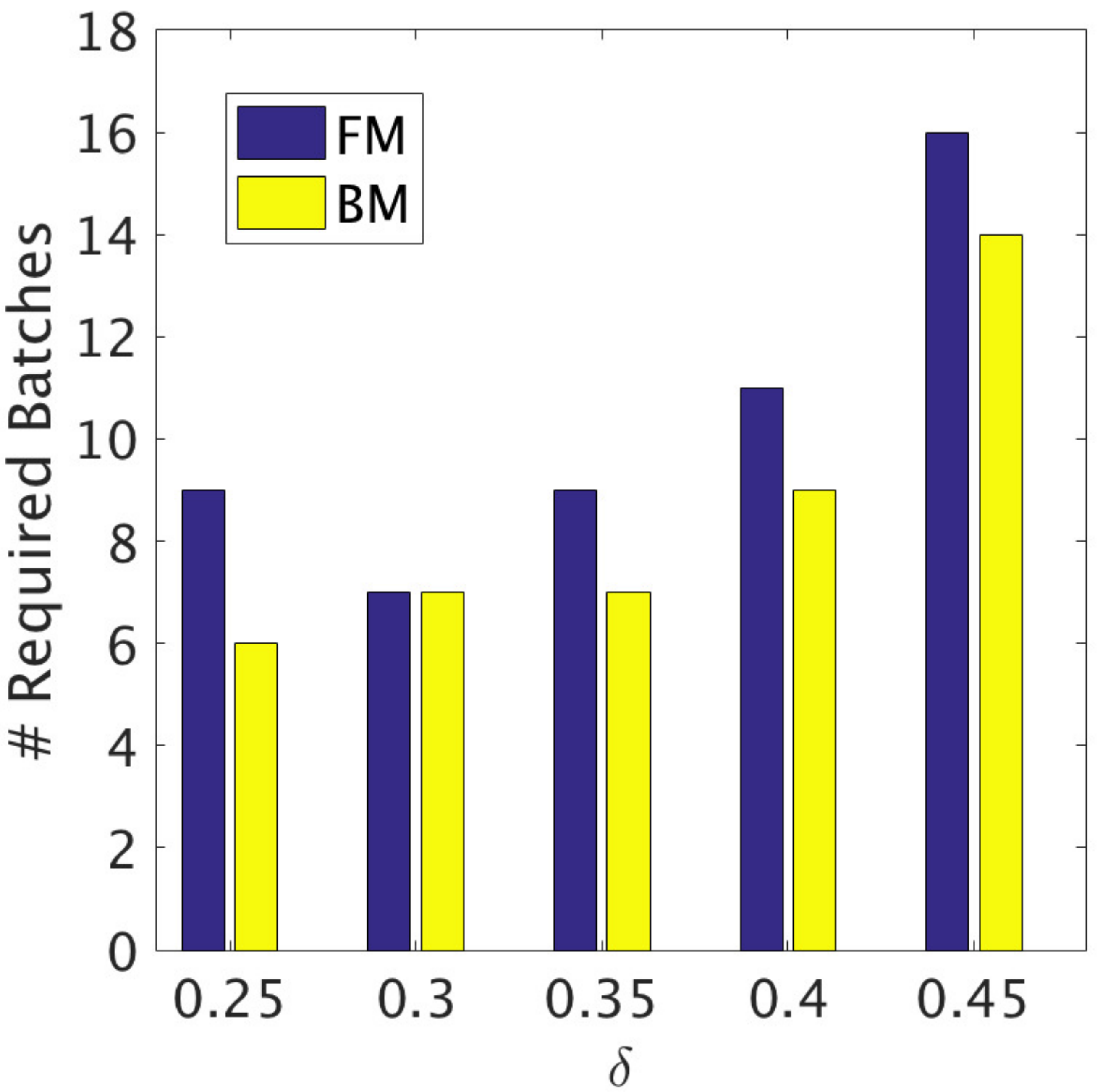}
	}\vspace{-2ex}
	\caption{The Effects of $m$ and $\delta$ on Real Data}\vspace{-3ex}
\end{figure*}
\vspace{-1em}

\noindent\textbf{Accuracy.} To compare the accuracy of the five methods, the number of questions $m$ varies from 250 to 1000, $\delta$ varies from 0.2 to 0.6, and $\lambda=1$. The results are shown in Figure \ref{fig:syn_m_accuracy_0.2}-\ref{fig:syn_m_accuracy_0.6}. From the above five figures, we can observe that for each possible value of $\delta$, our methods FM and BM generally maintain a high accuracy within 95$\%$-100$\%$ even when the quantity of questions is small ($m =250$); Furthermore, our approaches have outperformed Qasca whose highest accuracy is less than 95$\%$. Although FM+TF, BM+TF have high accuracy above 90$\%$, their highest accuracy is lower than ours for most cases, especially when $\delta=0.3$ and $0.4$. There is small increase of accuracy when $m$ becomes larger.

Besides, as $\delta$ increases, the accuracy of all methods regarding the same $m$ (number of questions) is improved, which exactly proves that $\delta$ controls the quality of the returned answers. A higher value of $\delta$ indicates a higher accuracy of the returned answers.

\noindent\textbf{Number of required batches.} Next, we study the latency of each method depicted in Figure \ref{fig:syn_m_batches_0.2}-\ref{fig:syn_m_batches_0.6} that measured in number of batches. As $\delta$ increases, the number of required batches of all methods regarding the same $m$ (the number of questions) is increased. From the analysis above, we know that $\delta$ controls the quality of the returned answers and a higher value of $\delta$ indicates that the returned answers have higher accuracy. In order to maintain an overall higher accuracy of questions, all methods choose to assign more repetitions of questions to workers and more certain answers can be obtained, which results in a longer latency (the number of required batches).

More importantly, when $\delta$ is ranging from 0.2 to 0.5 (Figure \ref{fig:syn_m_batches_0.2}-\ref{fig:syn_m_batches_0.5}), the latency of FM and BM is smaller than that of Qasca, FM+TF and BM+TF for different numbers of questions for most cases (sometimes they have similar values). Combining Figure \ref{fig:syn_m_accuracy_0.2}-\ref{fig:syn_m_accuracy_0.5} where the accuracy of our methods is better than those of others, we can conclude that when $\delta$ is ranging from 0.2 to 0.5, FM and BM are able to balance the latency and accuracy and they perform better than Qasca, FM+TF and BM+TF. In other words, the accuracy is improved without sacrificing the latency, namely, assigning more repetitions to workers to increase the quality. When $\delta=0.6$, FM and BM have longer latency than that of other methods (Figure \ref{fig:syn_m_batches_0.6}), but the accuracy is still the highest (Figure \ref{fig:syn_m_accuracy_0.6}). This also reveals that when $\delta$ has a high value like 0.6, the increase of accuracy for FM and BM tends to be smaller while the latency may be larger than that of others. Therefore, in latter experiments, we set $\delta=0.3$ as the default value (since the accuracy is already high, and starting from it, the increase of accuracy becomes smaller but the overall latency will become larger), and $rep=3$ accordingly. Besides, with the increase of  the number $m$ of questions, the growth of required batches becomes slower, which shows that our methods are effective to apply to question pools of large quantities.

\noindent\textbf{The effect of value $\delta$.} Note that some observations related to $\delta$ are discussed above. In this experiment, to better compare the performance of FM and BM, we set $m=500$ and vary $\delta$ from 0.2 to 0.6. The results are shown in Figure \ref{fig:syn_delta_accuracy} and \ref{fig:syn_delta_batches}. From Figure \ref{fig:syn_delta_accuracy}, as we increase the standard of returned questions, the easiness score threshold $\delta$ is enlarged, BM dominates FM in terms of the accuracy in most cases. The reason is that BM always chooses the best worker for a question, whereas FM simply assigns a worker to the first open question. By checking Figure \ref{fig:syn_delta_batches}, we find that BM maintains a similar latency to FM while achieving higher accuracy than FM. Most importantly, this experiment proves that the setting of $\delta$ is useful and beneficial for the requesters to control the accuracy of returned answers.

\noindent\textbf{Steady state of the parameter estimation model.} In this experiment, we let $m=500$ and $\delta=0.3$, and we aim at studying the speed of the parameter estimation model to reach the steady state (convergence). The left and right cycles of the parameter estimation model are executed for 5 iterations respectively, and then this process is repeated for 20 iterations. The confidence of all the answers are tracked and the average variation in each iteration are drawn in Figure \ref{fig:r_conv}.  From the figure, we can observe that the changes of the answer confidence decline at a fast speed for the first 5 iterations and extremely approach 0 starting from the 7th iteration. This discovery shows the effectiveness of the model and guides us to set an appropriate value for the $\#$ iterations. We simply fix the number of external iterations to 10 which is enough for us to get a stable answer confidence and control the latency caused by transitions between batches.
\begin{table}[t]
	\centering
	\begin{footnotesize}
		\caption{Parameter Settings for Experiments on Real Data}
		\label{real_paras}{\small
		\bigskip\vspace{-4ex}
		\begin{tabular}{l|l|l}
			Parameter & Description & Value\\
			\hline			\hline
			$m$ & $\#$ questions & 20, \textbf{30}, 40, 50, 60 \\
			$\delta$ & easiness score threshold & 0.25, \textbf{0.3}, 0.35, 0.4, 0.45 \\
			$b$ &$\#$ total batches &\textbf{20}, 200, 2K, 20K\\
			$rep$ & $\#$ repetitions in qasca & 3\\
			$k$ & top-$k$ tasks in qasca & 1\\
			\hline
		\end{tabular}}
	\end{footnotesize}\vspace{-4ex}
\end{table}\vspace{-2ex}

\subsection{Experiments on Real Data}
\subsubsection{Experimental Settings}
First of all, we extract the tuples with each containing one worker and a list of her question submission times from the CrowdFlower dataset. Since the time span of the whole dataset lasts for 20 days and the questions are simple, the worker processing times are neglected. Therefore, the question submission times are approximated as the worker arrival times. For 39 workers in the Duck Identification dataset, we randomly match each of them with one worker in the CrowdFlower dataset to monitor their arrival times. This mapping method is also adopted in \cite{Cheng:2017:URR:3035918.3064008}. In order to yield persuasive results, for FM, BM and Qasca, we repeat the random matching process and conduct question assignment for 100 times, and then the accuracy and the number of required batches are averaged. After the mapping is finished, the whole time span of the workers is evenly divided into 20 batches by default.

The parameter settings are shown in Table \ref{real_paras}. In addition, settings of the parameter estimation model are exactly the same as those in the synthetic experiments. Furthermore, $\delta$ (easiness score threshold) is from 0.25 to 0.45, and 0.3 is picked as the default value (the reason is indicated in the first experiment on synthetic dataset); $rep=3$ and $k=1$ for Qasca; $m$ (number of questions) varies from 20 to 100 and the default value is 40.

\begin{figure}[htbp]
	\centering
	\subfigure[Accuracy]{\includegraphics[height = 1.4in,width=0.48\linewidth]{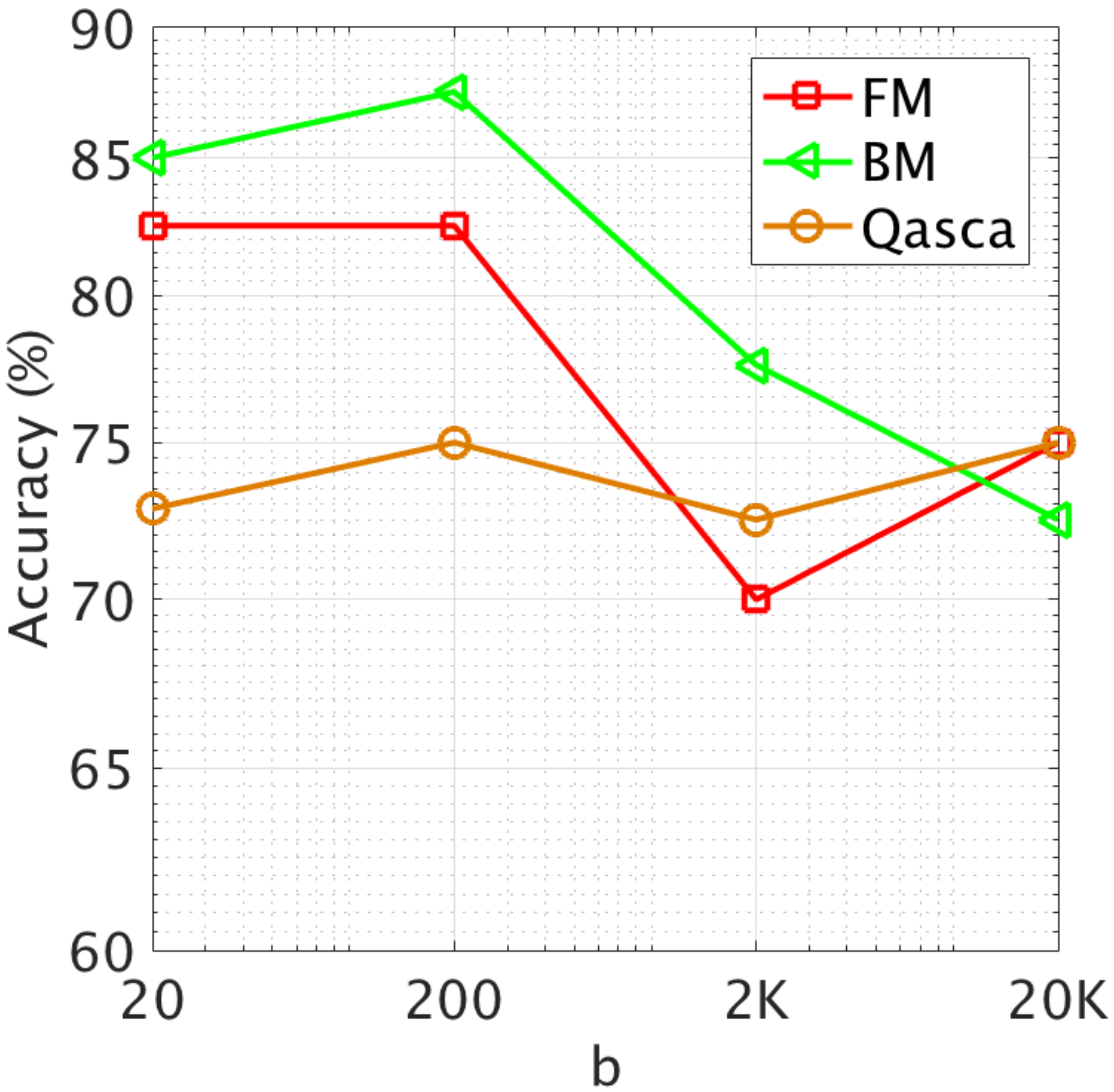}
		\label{fig:r_b_accuracy}}
	\subfigure[$\#$ Required Batches]{\includegraphics[height = 1.4in,width=0.48\linewidth]{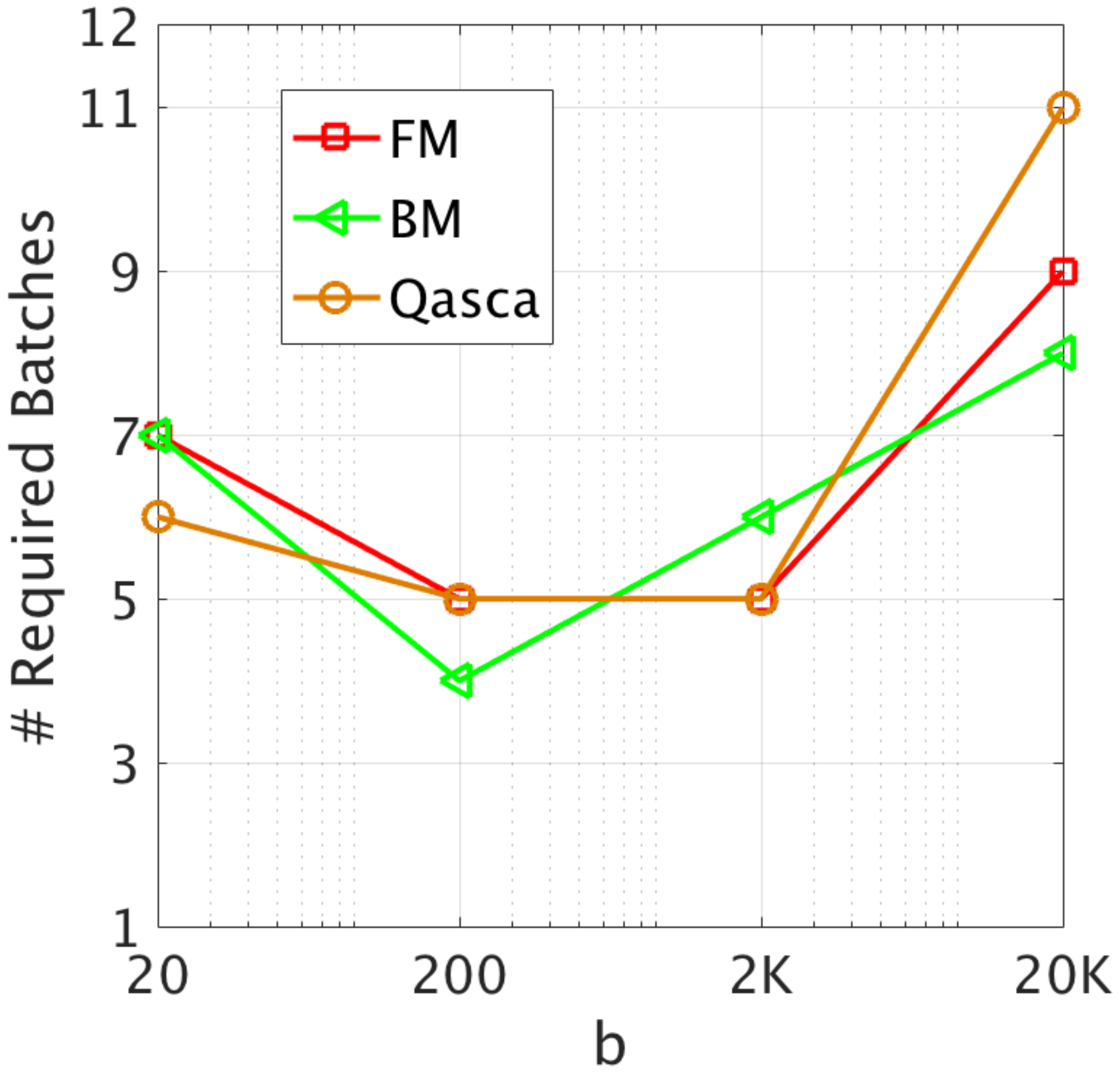}
		\label{fig:r_b_batches}}\vspace{-2ex}
	\caption{ The Effect of $b$ on Real Data}\vspace{-3ex}
	\label{fig:r_b}
\end{figure}

\subsubsection{Experimental Results} \hfill 
\vspace{-2ex}
\textbf{Accuracy and number of required batches.} To evaluate the accuracy and the number of required batches for different numbers of questions, the numbers of questions are set from 20 to 60 and the easiness score threshold $\delta$ is 0.3. From the results in Figure \ref{fig:r_m_accuracy}, FM and BM have higher accuracy than that of Qasca. Furthermore, the number of required batches (latency) are recorded in Figure \ref{fig:r_m_batches}. It is noticed that the latency increases as more questions are released for the workers. And FM and BM achieve smaller overall latency than Qasca for most of the time. These two experiments (combining the experiment in Figure \ref{fig:r_m_accuracy}) together prove that our MCQ approaches have the ability to balance the latency and the accuracy on real data, namely, the resulted high accuracy is not caused by assigning more repetitions to workers.

\noindent\textbf{The effect of value $\delta$.} In this experiment, we set $m=40$ and vary $\delta$ from 0.25 to 0.45. From Figure \ref{fig:r_delta_accuracy}, as $\delta$ becomes larger, the accuracy of FM and BM are generally increasing. The reason is that as $\delta$ increases, an answer can only be returned when reaching a higher confidence, which provides a guarantee of quality for the returned answers. Combining the observation in Figure \ref{fig:r_delta_batches} where BM completes all questions in a smaller quantity of batches for most cases, BM performs better than FM. Therefore, we conclude that as $\delta$ increases, BM becomes more advantageous to provide higher accuracy and better latency than FM.

\noindent\textbf{The effect of value $b$.} Next, we investigate the effect of value $b$, namely, the number of total batches. $b$ is set to 20 by default but in this experiment, it is set to 4 different values from 20 to 20K. The results are shown in Figure \ref{fig:r_b}. Note that $b$ is the number of overall batches configured by the requesters while \textit{$\#$ required batches} refers to the overall batches needed to complete all questions. In Figure \ref{fig:r_b_accuracy} and Figure \ref{fig:r_b_batches}, compared to FM and Qasca, BM achieves the highest accuracy and shortest latency for most of the cases. Furthermore, as $b$ increases, the required batches of FM and BM are increased and their accuracy also drops. This problem is not caused by our approaches. It can be explained by the fact that the average number of workers in a batch will be extremely small or even equal to 0 if we divide the time span into substantial batches. In this case, our approaches are similar to the online scenarios since the joint benefits of workers cannot be considered, which results in a less satisfying accuracy. What's more, since the knowledge of workers is hard to obtain for the initial batches, the answer confidence will take a longer time to reach a steady state, thus the delay is enlarged. Therefore, how to choose a proper batch interval is very important.

\subsection{Summary}
We finally summarize our experimental findings as follows.

\begin{itemize}[leftmargin=*]
	\item{} In most experiments on synthetic and real datasets, our proposed MCQ solutions FM and BM have better performance than the baseline Qasca in terms of the accuracy and overall latency. We conjecture that our solutions in the synchronized task assignment can improve the accuracy of answers while maintaining an acceptable latency.
	\item{} The BM solution assigns workers to questions with the largest increase of easiness score and it outperforms FM in most cases.
	\item{} The iterative parameter estimation model has a fast speed to reach a steady state. Furthermore, it performs better than TRUTHFINDER framework in the synchronized task assignment scenarios.
	\item{} The introduction of easiness score threshold $\delta$ effectively controls the accuracy of the returned answers. 
\end{itemize}

%% file: Relatedwork_v1.tex
\vspace{-2ex}
\section{related work} \label{relatedwork}
\textbf{Parameter Estimation.} In crowdsourcing, parameter estimation is widely used to infer latent variables such as worker quality and task truths. For example, \cite{yin2008truth} proposes a so-called TRUTHFINDER framework to discover truths over conflicting Internet information, and parameters like worker expertise is iteratively computed until convergence; based on \cite{yin2008truth}, Dong et al. \cite{dong2009integrating} leverage Bayesian analysis to estimate the accuracy of a data source by taking the influence of source dependence into consideration. And more complicated relationships of dependence between workers are further discussed in \cite{dong2010global,dong2009truth}; \cite{whitehill2009whose,welinder2010multidimensional,karger2011iterative,fan2015icrowd,ipeirotis2010quality} devise their truth inference model based on an Expectation-Maximization (EM) method \cite{dawid1979maximum}. 

\noindent\textbf{Task Assignment} \cite{chittilappilly2016survey}\textbf{.} Recent work regarding online task assignment in general-purpose crowdsourcing mainly aims to provide a quality control on the achieved results \cite{boim2012asking,fan2015icrowd,khan2017crowddqs,liu2012cdas,zheng2015qasca}. For example, \cite{khan2017crowddqs} studies when to insert a gold standard question that helps estimate the worker accuracy and supports blocking of poor workers; \cite{boim2012asking} decides the task direction by minimizing the uncertainty of the collected data; \cite{liu2012cdas} estimates the workers' accuracy according to their previous performance and the core quality-sensitive model is able to control the processing latency; \cite{zheng2015qasca} is the state-of-the-art work and it optimizes the evaluation metrics such as F-score on task selection and improves the previous work \cite{boim2012asking,liu2012cdas}. In addition, there are other online approaches. Yuan et al. \cite{yuen2015taskrec} devise a task recommendation framework based on the unified probabilistic matrix factorization; \cite{ho2012online} targets at a set of heterogeneous tasks and skill sets of workers are inspected over the tasks they have done. 

Except for online task assignment, offline assignment strategies are also fully investigated in \cite{Zhao:2013:TLB:2487575.2487708,zhao2015crowd,karger2011iterative}, but they are less relevant to our work. To be specific, the parameters such as worker quality in our work are adaptively updated in batches without assuming prior distributions, whereas offline work like \cite{Zhao:2013:TLB:2487575.2487708,zhao2015crowd} retrieve them from the historical performance of workers in crowdsourcing platforms like AMT. Furthermore, the offline assignment is one shot, which cannot handle the case that worker arrives or leaves halfway after the questions are assigned. 

Although TRUTHFINDER \cite{yin2008truth} has a good performance to infer website quality and truth of web information, it is run on aggregated facts of large quantity, whereas a crowdsourcing task has a few votes. Applying  TRUTHFINDER to crowdsourcing will result in slow convergence and inaccurate answer confidence. In order to address this issue, question easiness is integrated into the framework that accelerates the inference and assists in the task assignment optimization. With the insertion of question easiness, the whole model is adapted to better fit in crowdsourcing circumstances. Furthermore, synchronized task assignment is popular nowadays \cite{Zhang:2017}. Therefore, unlike most of the online task assignment that selects the top-k beneficial questions \cite{zheng2015qasca} to an individual worker, this work considers the synchronized task assignment scenario in which collaborative worker labor is efficiently utilized and the number of completed questions within each batch is maximized. In addition, compared to some offline task assignment work \cite{Zhao:2013:TLB:2487575.2487708,zhao2015crowd}, our model requires no prior knowledge of worker behaviors and is practical to apply on open crowdsourcing platforms where workers are quite dynamic.

%% file: Conclusion.tex
\vspace{-2ex}
\section{Conclusion} \label{conclusion}
In this paper, we formally propose the MCQ problem in synchronized task assignment scenario, which assigns questions to workers to make the most efficient utilization of worker labor and maximize the number of completed questions in each batch. The problem is proved to be NP-hard and two greedy approximation solutions are proposed to address the problem. Furthermore, we develop an efficient parameter estimation model to assist in the task assignment optimization. Extensive experiments on synthetic and real datasets are conducted to evaluate our approaches. The experimental results show that our approaches outperform the baseline and achieve a high accuracy while maintaining an acceptable latency.